\documentclass[journal,one column, draftcls, letter, 12pt]{IEEEtran}
\usepackage{amscd,graphicx,latexsym,multicol,cite,ulem,epsf}
\usepackage{verbatim}
\usepackage{graphics}
\usepackage{graphicx, subfigure}
\usepackage{amssymb}
\usepackage[cmex10]{amsmath}
\usepackage{amsthm} 
\usepackage{dsfont}
\usepackage{bbm}
\usepackage{setspace}
\usepackage{mathabx}
\usepackage{relsize}
\usepackage{breqn} 

\usepackage{accents}
\newcommand{\qeq}{\accentset{(a)}{=}}
\newcommand{\qeqb}{\accentset{(b)}{=}}

\DeclareMathOperator*{\argmax}{arg\,max}
\newcommand{\diag}{\mathop{\rm diag}}

\newtheorem{thm}{Theorem}
\newtheorem{conj}{Conjecture}
\newtheorem{cor}{Corollary}
\newtheorem{propos}{Proposition}

\newtheorem{mydef}{Definition}
\begin{document}
\title{Handoff Rate and Coverage Analysis in Multi-tier Heterogeneous Networks}
\author{Sanam Sadr and Raviraj S. Adve\\
University of Toronto \\
Dept. of Elec. and Comp. Eng.,  \\ 10 King's College Road, Toronto,
Ontario, Canada M5S 3G4\\
\{ssadr, rsadve\}@comm.utoronto.ca }
\date{}
\maketitle
\vspace*{-0.3in}
\begin{abstract}
This paper analyzes the impact of user mobility in multi-tier heterogeneous networks. We begin by obtaining the handoff rate for a mobile user in an irregular cellular network with the access point locations modeled as a homogeneous Poisson point process. The received signal-to-interference-ratio (SIR) distribution along with a chosen SIR threshold is then used to obtain the probability of coverage. To capture potential connection failures due to mobility, we assume that a fraction of handoffs result in such failures. Considering a multi-tier network with orthogonal spectrum allocation among tiers and the maximum biased average received power as the tier association metric, we derive the probability of coverage for two cases: 1) the user is stationary (i.e., handoffs do not occur, or the system is not sensitive to handoffs); 2) the user is mobile, and the system is sensitive to handoffs. We derive the optimal bias factors to maximize the coverage. We show that when the user is mobile, and the network is sensitive to handoffs, both the optimum tier association and the probability of coverage depend on the user's speed; a speed-dependent bias factor can then adjust the tier association to effectively improve the coverage, and hence system performance, in a fully-loaded network.\end{abstract}

\begin{IEEEkeywords}
Heterogeneous cellular networks, Poisson point processes, range expansion, user association, handoff rate, downlink coverage analysis.
\end{IEEEkeywords}

\section{Introduction}
It is now widely accepted that the next generation of wireless communication systems will be based on multi-tier heterogeneous wireless access with an all-IP infrastructure~\cite{akyildiz:04,fabio:02}. With the anticipated increase in the number of applications available to a handheld device, e.g., voice, data, real-time multimedia,~etc.~\cite{kilper:11, cisco12:17}, mobility management will play an important role in providing seamless service to the mobile users moving from one access point (AP) to the other. Client-server applications such as email, web browsing, etc., are amenable to short-lived connections and do not require sophisticated mobility solutions. Media streams, however, can function normally only with a maximum interruption of 50msec; while an interruption of 200msec is still acceptable, any longer interruption causes perceptible and unacceptable delays~\cite{banerjee:03}. Therefore, on the one hand, it is desirable to minimize handoffs between APs to avoid any excessive connection delay and call drops; on the other hand, depending on the radio resource allocation scheme, a user entering the service area of another AP could cause (and suffer) increased interference levels if the handoff is not performed. Therefore, in the move towards multi-tier heterogeneous networks, the issues of cell association and handoff must be addressed in an efficient manner to minimize the service delay and connection failure for mobile users.

\subsection{Related Work}
In the context of cellular networks, there is a large body of literature studying the delays caused due to handoffs~\cite{banerjee:03,nakajima:03,Lin:07}, protocols and effective handoff algorithms~\cite{kim:96, iera:02, akyildiz:04, ekiz:05}, and multi-tier system design with microcells overlayed by macrocells~\cite{ganz:97, lagrange:97}. If there are enough resources, the classic handoff algorithms in a multi-tier network assign users to the lowest tier (e.g., the microcells) thereby increasing system capacity~\cite{iera:02}. To account for mobility, based on an estimated sojourn time compared to a threshold, the user is classified as slow or fast, and is assigned to the lower or the upper tier respectively~\cite{Beraldi:96, Zeghlacke:96} (these works assume a two-tier network). The estimated sojourn time depends on the cell dimensions as well as user information such as the point of entry and user trajectory~\cite{yeung:96}. Similarly, velocity adaptive algorithms use the mobility vector, including both the estimated velocity and the direction, to perform the handoff~\cite{kim:96}. To avoid the ping-pong effect due to unnecessary inter-tier handoffs, once the user is classified as fast, it remains connected to the upper tier regardless of any changes in its speed~\cite{lagrange:96}. Another alternative is to introduce a dwell-time threshold to take into account the history of the user before any handoff decision~\cite{iera:02}. This technique is based on speed estimation at each cell border. 

Whether the handoff is performed solely by the network controller~\cite{3gpp:12}, or autonomous decisions by the user equipment are taken into account~\cite{pedersen:13}, it is desirable to reduce the signalling overhead due to unnecessary or frequent handoffs between the tiers or among the APs within one tier. The proposed algorithms mentioned above are mainly applicable in large cells. Importantly, the handoff rate, sojourn time or dwell time analysis provided in the literature consider deterministic AP locations and a regular grid for the positions of the base stations. With the increasing deployment of multi-tier networks, especially small cells in an irregular, non-deterministic manner~\cite{ghosh:12}, handoff analyses for modern heterogeneous networks must now take into account the randomness of the AP locations by using random spatial models~\cite{andrewsSpatial:10}.

One of the common 2-D spatial models, with the advantage of analytical tractability, is the homogeneous Poisson point process (PPP) characterized by the density of nodes, $\lambda$. In this model, the number of nodes in area $A$ is a Poisson random variable with mean $\lambda A$; the number of nodes in disjoint regions are independent random variables, and node locations are mutually independent. The accuracy of this model for a two-tier cellular network was examined by Dhillon \textit{et al}.~\cite{dhillonJournal:12}; it was shown that the probability of coverage in a real-world 4G network lies between that predicted by the PPP model (pessimistic lower bound) and that by the regular hexagonal grid model (optimistic upper bound) with the same AP density. Similar results were reported comparing the coverage predictions by the PPP and the square grid model~\cite{andrews:11}. Furthermore, the PPP model provides a tighter bound at cell edges where the probability of having a dominant interferer is closer to that found in an actual 4G network~\cite{andrews:11}. Wen \textit{et al.}~\cite{Wen:14} extended the analysis in~\cite{dhillonJournal:12}, and derived the probability of coverage and the average user throughput in a multi-hop multi-tier network. In both works, the maximum received signal-to-interference-plus-noise-ratio (SINR) at the user is used as the tier connection metric. The user is then in coverage if its received SINR is above the tier's threshold either through direct connection to an AP (infrastructure mode) or through another user (ad hoc mode)~\cite{Wen:14}. 

Modeling each tier of a multi-tier network by an independent PPP, Jo \textit{et al.}~\cite{shin:12} derived the outage probability and the ergodic rate of a multi-tier network with flexible tier association. A popular method to achieve this flexibility is by adding, in dB, a bias factor either to the average received power~\cite{shin:12,Mukherjee:ICC12} or to the instantaneous received SINR~\cite{Mukherjee:ICC12}, and use this biased received power or biased SINR as the tier connection metric. If all the APs in the same tier have the same transmit power (an assumption common to most works), connecting to the AP with the maximum average received power results in the 2-D space tessellation in which the AP coverage area is represented by a Voronoi cell. The cell associated with an AP then comprises those points of space that are closest to the AP; in turn, the cell size is a continuous random variable. An analytical approximation for the Voronoi cell size has been derived in~\cite{Ferenc:07}. Applying this formula to a network with the user distribution modelled by an independent PPP, the distribution of the AP load in terms of the number of users is derived in a single tier~\cite{Yu:13}, and a multi-tier network with flexible tier association~\cite{Singh:13}. The authors in~\cite{Yu:13} obtained the probability that an AP is inactive and used it to derive the probability of SINR coverage taking into account only the received signals from the active APs. Using the proportionally fair model where each AP equally divides its available bandwidth amongst its users, the load distribution was used to derive the downlink rate distribution at a reference user in a multi-tier network~\cite{Singh:13}, and a two-tier network with limited backhaul capacity~\cite{Singh:14}. 

The effect of the dominant interferer was studied by Heath \textit{et al.}~\cite{Heath:13} to derive a tractable model for the total interference at a specific cell in a heterogeneous network. In this model, choosing a fixed-size cell (as the cell under consideration) and a guard radius (hence a guard region), the interfering APs form a PPP and consist of all APs lying outside the guard region with the nearest one as the dominant interferer. This work allows us to evaluate the performance for a ``given" cell as apposed to a ``typical" cell in the entire network. Clearly, in a multi-tier network, load distribution through tier association affects the number of users serviced by each tier. To deal with the resulting inter-tier and intra-tier interference, two main approaches have been proposed: 1) spectrum sharing among tiers but with fractional frequency reuse~\cite{novlan:J12}; 2) orthogonal spectrum allocation among tiers, thereby eliminating inter-tier interference, e.g.,\cite{bao:J14, SadrAdveCL:14}. The relation between the allocated fraction of spectrum and the tier association probability to maximize the probability that a typical user receives its required rate has been studied in~\cite{SadrAdveCL:14}. It is shown that in the case where an average number of users per AP (only controlled by the relative distribution densities and the bias) is assumed, the intuitive solution that the fraction of spectrum allocated to each tier should be equal to the tier association probability results in essentially zero performance loss.

In all the works mentioned above in the context of heterogeneous networks modelled by PPP, the performance metric is evaluated at a stationary (but randomly located) user. In this paper, we focus on \textit{mobility-aware} analyses. Since interference management and tier association are coupled, handoff analysis becomes important in deriving effective mobility-aware solutions that can improve not only the service delay but the overall system capacity through load distribution in the network. A comprehensive survey of mobility models and their characteristics can be found in~\cite{Tracy:02, Roy:11} and the references therein. They are generally categorized into trace-based and synthetic models, and are mainly used to simulate the movement of the mobile users in ad hoc networks as realistically as possible. 

To the best of our knowledge, the only work that applied a mobility model in a PPP network, and derived an analytical expression for the handoff rate and sojourn time is by Lin \textit{et al.}~\cite{Lin:13}. The authors proposed a modified random waypoint (RWP) model\footnote{The RWP model~\cite{Johnson:96} is one of the most commonly used mobility models for evaluating the performance of a protocol in ad hoc networks. In this model, each node picks a random destination uniformly distributed within an underlying physical space, and travels with a speed uniformly chosen from an interval. Upon reaching the destination, the process repeats itself (possibly after a random pause time).} in a single-tier irregular network with APs modelled by a PPP, and show that the handover rate is proportional to the square root of the AP density. This work defines the handoff rate as the ratio of the average number of cells a mobile user traverses to the average transition time (including the pause time). Their analysis predicts a slightly higher handover rate and lower sojourn time (overall, a pessimistic prediction) compared to an actual 4G network. The handover rate and sojourn time predictions in this work, along with the coverage predictions in~\cite{dhillonJournal:12,andrews:11}, imply that the PPP model provides a slightly pessimistic but sufficiently accurate analysis while being analytically tractable. 

\subsection{Our Contributions}
Attracted by its applicability and tractability, we use the PPP model for handoff analysis in an irregular multi-tier network. Our goal is to analyze the impact of mobility, and to use this analysis in deriving effective tier association rules and, hence, load distribution in a multi-tier network to minimize this (negative) impact. Similar to~\cite{Lin:13}, we consider the handoff rate during one movement period. However, we use a different mobility model (as apposed to the modified RWP) and a different definition for handoff rate that includes the connection metric. Specifically, we define the handoff rate as the \textit{probability} that the user crosses over to the next cell in one movement period. The contributions of our work are as follows:
\begin{itemize}
\item First, we derive the handoff rate in a network where the AP locations are modeled by a homogeneous PPP. This handoff rate can also be interpreted as the probability that the serving AP does \textit{not} remain the best candidate in one movement period. Based on some mild approximations, we simplify this expression; our numerical simulations show that our theoretical expression provides reliable results over a broad range of system parameters. 

\item Next, we consider coverage in a network with mobile users, assuming that a certain fraction of handoffs result in connection failure, i.e., we assume that the outage probability is linearly related to the handoff rate derived earlier. We use the signal-to-interference-ratio (SIR) in an interference-limited network and a pre-specified SIR threshold to define coverage at a reference user. We derive the probability that a mobile user, initially in coverage, remains so despite its motion. Interestingly, the results show that the degradation in service - even for a fast moving user in a network where most handoffs result in outage - decreases with the increase in the SIR threshold (provided that the user is allocated its desired resources). This approach provides a tractable model to analyze the impact of mobility; specifically, we \emph{do not} attempt to derive a joint coverage probability distribution across the locations of a specific mobile user. 

\item Finally, we consider a multi-tier network with orthogonal spectrum allocation among the tiers, and derive the coverage probability considering handoffs. We show that the overall network coverage can be improved by adjusting the tier association through the corresponding bias factor in a mobility-aware manner, hence, improving system performance in a fully-loaded network. This multi-tier coverage expression with handoff can easily be generalized to include spectrum sharing across tiers using~\cite{shin:12}. 
\end{itemize}

The remainder of the paper is organized as follows. In Section~\ref{system}, we describe the PPP model, the multi-tier network under consideration, and the tier association metric. The handoff rate for a typical user in such a network is derived in Section~\ref{sec:hoffrate}. This is the main result of the paper upon which the subsequent sections rely. The probability of coverage in a single tier considering a linear cost function for the handoff is derived in Section~\ref{coverage1}. Section~\ref{coverage2} presents the probability of coverage in a multi-tier network with and without the handoff analysis. We show how the bias factor can optimize the tier association for a mobile user to maximize the probability of coverage in a multi-tier network. We conclude the paper in Section~\ref{conclusions}.

Throughout the paper, we use $\mathds{P}(Y)$ to denote the probability of event $Y$ and $\mathds{E}_{X}(\cdot)$ to denote the statistical average with respect to the random variable $X$.

\section{System Model} \label{system}
We consider a network comprising $K$ tiers in the downlink. Each tier, indexed by $k$, is characterized as a homogeneous Poisson point process $\Phi_{k}$ with the tuple $\{P_{k},\lambda_{k}, \tau_{k}\}$ denoting the transmit power, the AP density and the SIR threshold respectively. The tiers are organized in increasing order of density i.e., $\lambda_{1} \leq \lambda_{2} \cdots \leq \lambda_{K}$. Given the density $\lambda_{k}$, the number of APs belonging to tier $k$ in area $A$ is a Poisson random variable, with mean $A\lambda_{k}$, which is independent of other tiers. Furthermore, all the APs in tier $k$ have the same transmit power $P_{k}$.

In our model, transmissions suffer from a constant path loss exponent $\alpha$. Rayleigh fading with unit average power models the channel gain for the received signal from each AP. Since networks such as those under consideration here are interference-limited, we ignore thermal noise, and use the received SIR as the coverage metric. Finally, for tractability, we ignore shadowing. In the coverage analysis, we assume orthogonal spectrum allocation among tiers and a reuse factor of one within each tier. Therefore, at a typical user connected to tier $k$, the set of interfering APs include all the APs in tier $k$ except the serving AP\footnote{A reuse factor of greater than one can be accounted for by using a reduced AP density in calculating the interference~\cite{andrews:11}, but will not change the handoff rate.}. The results derived in this paper can be generalized to allow for spectrum sharing across tiers~\cite{shin:12}, log-normal shadowing~\cite{Dhillon:14}, and any arbitrary fading distribution for the interfering signals~\cite{andrews:11}.

To evaluate the received signal at a typical user, we shift all the point processes so that the user under consideration lies at the origin. Despite the shift, each tier still forms a PPP with the same density. We use the maximum ``biased" average received power as the tier association metric where the received power from all the APs of different tiers are multiplied by the corresponding bias factor $B_{k}$, and the user is associated with the tier with the maximum biased average received power. Later in the paper, we will optimize the bias factor to account for mobility. Let $r_{j}$ denote the distance between a typical user and the nearest AP in the $j$th tier. In this setup, the user connects to tier $k$ if:
\begin{equation}
\label{eq:LayerAsso}
k = \argmax_{j \in \{1,\cdots,K\}} P_{j}L_{0}(r_{j}/r_{0})^{-\alpha}B_{j},
\end{equation}
where $B_{j}$ is the bias factor associated with tier $j$, $L_{0}$ is the path loss at reference distance $r_{0}$, and $\alpha$ is the path loss exponent. We use $r_{0} = 1$m and $L_{0} = (4\pi/\epsilon)^{-2}$ where $\epsilon$ denotes the wavelength at 2GHz. Since all the APs in each tier have the same transmit power and bias factor, the best candidate from each tier is the AP closest to the user. Without loss of generality, we set $B_{1} = 1$. If all the tiers have the same bias factor (or simply, $B_{j} = 1, \; \forall j$), the tier association metric is the maximum received power. Otherwise, $B_{j} > 1$ results in an increased number of users connecting to tier $j$.

Within the serving tier, the user connects to the nearest AP in that tier. If the user is initially in coverage, when it moves, it might fall into the coverage area of another AP at a shorter distance, and a handoff occurs. Although the user might be in coverage at both locations, rapid changes in the serving AP increases the possibility of connection failure. Fig.~\ref{fig:MobMod} shows the scenario under consideration. $l_{1}$ is the user's initial location at connection distance $r$ from the serving AP denoted by $AP_{s}$. The user moves a distance $v$ in a unit of time, at angle $\theta$ with respect to the direction of the connection, to a new location $l_{2}$ at distance $R$ from $AP_{s}$. This model is mostly suitable for a scenario where the user moves at a constant speed or it has small variation such that it can be approximated by its mean.

\section{Handoff Rate in a Single Tier Network}
\label{sec:hoffrate}

Whether the handoff occurs (Fig.~\ref{fig:MobMod2}) or not (Fig.~\ref{fig:MobMod1}) depends on the existence of another AP in the circle with the user at the center and radius $R$.
\begin{figure}
     \begin{center}
        \subfigure[Scenario where the handoff occurs.]{%
            \label{fig:MobMod2}
            \includegraphics[width=0.5\textwidth]{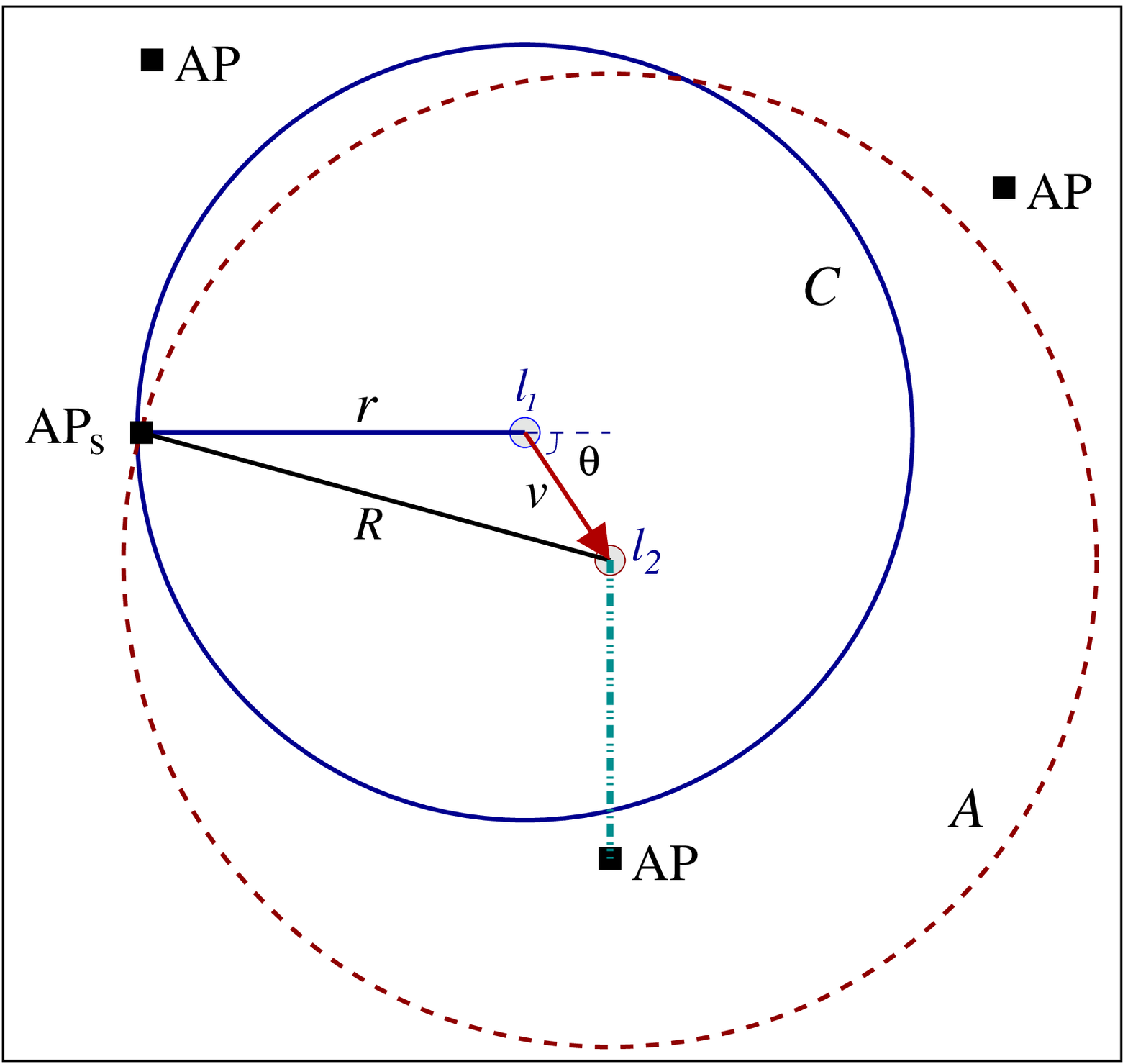}
        }%
        \subfigure[Scenario where the handoff does not occur.]{%
           \label{fig:MobMod1}
           \includegraphics[width=0.5\textwidth]{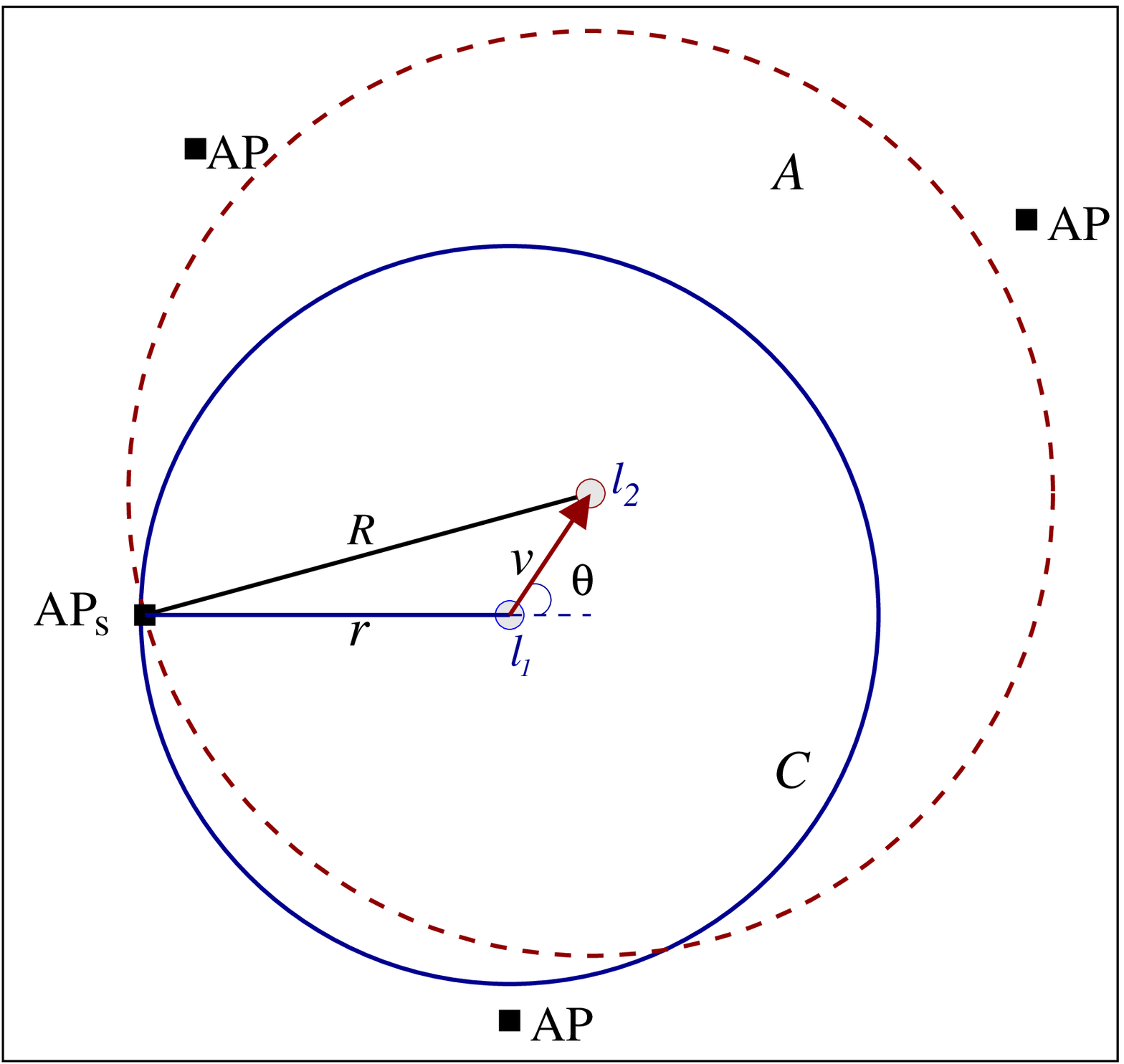}
        } 
        \end{center}
    \caption{Scenario where the user is initially at $l_{1}$, at connection distance $r$ from the serving AP, moving a distance $v$ in the unit of time at angle $\theta$ with the direction of the connection; (a) handoff occurs if there is another AP closer than $R$ to the user at the new location $l_{2}$; (b) the serving AP remains the closest AP to the user at location $l_{2}$. Hence, handoff does not occur.}
   \label{fig:MobMod}
\end{figure}
In Fig.~\ref{fig:MobMod}, $\mathcal{C}$ denotes the circle with its center at $l_{1}$ and radius $r$; $\mathcal{A}$ denotes the circle with its center at $l_{2}$ and radius $R$. The two circles intersect in at least one point which is $AP_{s}$. The excess area swiped by the user moving from $l_{1}$ to $l_{2}$ is denoted by $\mathcal{A} \setminus \mathcal{A} \cap \mathcal{C}$. For the most part, we assume that the user can move in any direction with equal probability. Later in the paper, we show that due to symmetry, the probability of handoff for the user moving at angle $(2\pi - \theta)$ is the same as that for the user moving at angle $\theta$ with the direction of the connection. Denoting the corresponding random variable as $\Theta$, the probability distribution function (PDF) of $\Theta$ is then set to be non-zero in $[0,\pi)$; here, we assume a uniform distribution given by $f_{\Theta}(\theta) = 1/\pi$.

Let $H_{k}$ denote the event that a handoff occurs for a user connected to tier $k$. Throughout the paper, we denote the complementary event that a handoff does not occur for the user connected to tier $k$ as $\widebar{H}_{k}$. Furthermore, let random variable $R_{k}$ denote the distance between the user and the closest AP in tier $k$. 

\begin{mydef}
The handoff rate, $\mathcal{H}_{k}(v, \lambda_{k}) = \mathds{P}(H_{k})$, denotes the probability of handoff for a user connected to an AP belonging to tier $k$, moving a distance $v$ in a unit of time (speed of $v$).
    \end{mydef}
The handoff rate, $\mathcal{H}_{k}(v, \lambda_{k})$, is a function of $v$ and $\lambda_{k}$ and is given by Theorem \ref{thrm:handoffcondition}.
\begin{thm} 
\label{thrm:handoffcondition}
Consider a mobile user at connection distance $r$ in a network with access points distributed according to a homogeneous PPP with density $\lambda_{k}$. The probability of handoff $\mathds{P}(H_{k}|r, \theta)$ for the user moving a distance $v$ in a unit of time at angle $\theta$ with respect to the direction of the connection is given by:
\begin{equation}
\label{eq:handoffcondition}
\mathds{P}(H_{k} | r, \theta) = 1- \exp \Bigg[ -\lambda_{k} \Bigg(R^{2} \left[ \pi - \theta + \sin^{-1}\left(\frac{v\sin\theta}{R}\right) \right] - r^{2} (\pi - \theta) + rv\sin \theta \Bigg) \Bigg],
\end{equation}
where $R = \sqrt{r^2 + v^2 + 2rv\cos \theta}$. Furthermore, in such a network, the handoff rate for a uniformly distributed $\theta$ is given by:
\begin{multline}
\mathcal{H}_{k}(v, \lambda_{k}) = \\
1 - \frac{1}{\pi} \int_{\theta = 0}^{\pi} \int_{r=0}^{\infty} 2 \pi \lambda_{k} r \exp \Bigg[ -\lambda_{k} \Bigg(R^{2} \left[ \pi - \theta + \sin^{-1}\left(\frac{v\sin \theta}{R}\right) \right] + r^{2}\theta + rv\sin \theta \Bigg) \Bigg] \mathrm{d}r \mathrm{d}\theta.
\end{multline}
\end{thm}
\begin{proof}
See Appendix \ref{AppendixA}.
\end{proof}

For the special case $\theta = 0$ when the user is moving radially away from the serving AP, the handoff rate can be written in closed form and is given by the following corollary.
\begin{cor}
The handoff rate for the user moving radially away from the serving AP, i.e., $f_{\Theta}(\theta) = \delta(\theta)$ where $\delta(\cdot)$ is the Dirac delta function, is given by:
\begin{equation}
\mathcal{H}_{k}(v, \lambda_{k}) = 1 - \left( e^{-\lambda_{k} v^{2}\pi} - 2 v\pi\sqrt{\lambda_{k}} \cdot Q(v\sqrt{2 \pi \lambda_{k}}) \right),
\label{eq:specialcase}
\end{equation}
where $Q(x) = \frac{1}{\sqrt{2\pi}} \int_{x}^{\infty} e^{-t^{2}/2}\mathrm{d}t$.
\end{cor}
\begin{proof}
From~\eqref{eq:handoffcondition}, the handoff rate is given by:

\begin{dmath}
\mathds{P}(H_{k}) = 1 - \mathds{E}_{R_{k}} \Big[\mathds{P}(\widebar{H}_{k} | r, \theta \mbox{ = } 0) \Big]   \\
= 1 - \int_{r = 0}^{\infty}e^{ -\lambda_{k} \pi \left(R^{2} - r^{2}\right)} f_{R_{k}}\left(r\right)\mathrm{d}r\\
= 1 - \int_{r = 0}^{\infty}e^{ -\lambda_{k} \pi \left(R^{2} - r^{2}\right)} \cdot 2 \pi \lambda_{k} r \cdot e^{-\pi \lambda_{k}r^{2}} \mathrm{d}r \\
\qeq 1 - \int_{r = 0}^{\infty}e^{ -\lambda_{k} \pi (r + v)^{2}} \cdot 2 \pi \lambda_{k} r \mathrm{d}r \\
\qeqb 1 - \left[ e^{-\lambda_{k} v^{2}\pi} - 2 v \pi \sqrt{\lambda_{k}} \cdot \int_{v \sqrt{2 \pi \lambda_{k}}} ^{\infty} \frac{1}{\sqrt{2 \pi}}e^{-t^{2}/2}\mathrm{d}t
 \right] \\
= 1 - \left( e^{-\lambda_{k} v^{2}\pi} - 2 v \pi \sqrt{\lambda_{k}}\cdot Q(v \sqrt{2\pi\lambda_{k}})\right),
\end{dmath}
where $(a)$ follows from using $R^{2} = (r + v)^{2}$ when $\theta = 0$, and $(b)$ follows from the change of variable $t = \sqrt{2\pi\lambda_{k}}(r + v)$, giving the desired result.
\end{proof}

As is clear from~\eqref{eq:handoffcondition}, in the general case where $\theta$ is uniformly distributed, there is no closed-form expression for the handoff rate. However, assuming the user displacement is much smaller than the connection distance, $v\ll R$, the handoff rate for the general case can be further simplified as derived below.
\begin{cor}
\label{propos:handoffRateGen}
The handoff rate $\mathcal{H}_{k}(v, \lambda_{k})$ for a typical mobile user moving a distance $v$ in a unit of time in a network where $v\ll R$ such that $\frac{v \sin \theta}{R} \approx 0$ and the access points are distributed according to a homogeneous PPP with density $\lambda_{k}$ is given by:
\begin{equation}
\label{eq:handoffRateGen}
\mathcal{H}_{k}(v, \lambda_{k})  = \\
1 - \frac{1}{\pi}\int_{\theta = 0}^{\pi}  \left [ 1 - 2 b(v, \lambda_{k}, \theta)\sqrt{\pi}e^{b^{2}(v, \lambda_{k}, \theta)} Q(\sqrt{2}b(v, \lambda_{k}, \theta)) \right] \exp \left( -\lambda_{k} v^{2} (\pi - \theta) \right) \mathrm{d}\theta,
\end{equation}
where $b(v, \lambda_{k}, \theta) = \sqrt{\pi \lambda_{k}} \frac{va(\theta)}{2\pi}$, and $a(\theta) = 2\cos \theta (\pi - \theta) + \sin \theta$.
\end{cor}
\begin{proof} 
Under the assumption that the movement per unit time is much smaller than the connection distance, i.e., $\frac{v \sin \theta}{R} \approx 0$, the probability of handoff conditioned on $r$ and $\theta$ in~\eqref{eq:handoffcondition} in Theorem~\ref{thrm:handoffcondition} is simplified to:
\begin{dmath}
\label{eq:handoffApprox}
\mathds{P}(H_{k}|r, \theta) = 1- \exp \Bigg( -\lambda_{k} \Big[R^{2} (\pi - \theta)- r^{2} (\pi - \theta) + rv\sin \theta \Big] \Bigg) \\
= 1 - \exp \Bigg( -\lambda_{k} \Big[v^{2}(\pi - \theta) +  rv \left( 2\cos \theta(\pi - \theta) + \sin \theta\right) \Big] \Bigg) \\
= 1 - \exp \Bigg( -\lambda_{k} \Big[v^{2}(\pi - \theta) +  rv a(\theta) \Big] \Bigg),
\end{dmath}
where $a(\theta) = 2\cos \theta (\pi - \theta) + \sin \theta$. The handoff rate is then given by:

\begin{dmath} 
\label{HoffRate}
\mathds {P}(H_{k}) =  1 - \mathds{E}_{\Theta} \Big[ \mathds{E}_{R_{k}} \Big[\mathds{P}(\widebar{H}_{k} | r, \theta) \Big] \Big] \\
= 1 - \frac{1}{\pi} \int_{\theta= 0}^{\pi} \int_{r = 0}^{\infty} e^{ -\lambda_{k} \big(v^{2}(\pi - \theta) + rva(\theta) \big)} f_{R_{k}}\left(r\right) \mathrm{d}r \mathrm{d}\theta \\
= 1 - \frac{1}{\pi} \displaystyle \int_{\theta= 0}^{\pi} \int_{r= 0}^{\infty} e^{ -\lambda_{k} \big(v^{2}(\pi - \theta) + rva(\theta) \big) } \cdot 2\pi \lambda_{k} r e ^{- \pi \lambda_{k} r^{2}} \mathrm{d}r \mathrm{d}\theta \\
= 1 - \frac{1}{\pi} \displaystyle \int_{\theta= 0}^{\pi} e^{-\lambda_{k} v^{2}(\pi - \theta)} \int_{r= 0}^{\infty} e^{-\lambda_{k} \pi \Big(\left(r + \frac{va(\theta)}{2\pi} \right)^{2} - (\frac{va(\theta)}{2\pi})^{2} \Big)}2\pi \lambda_{k} r \mathrm{d}r \mathrm{d}\theta \\
= 1 - \frac{1}{\pi} \displaystyle \int_{\theta= 0}^{\pi} e^{-\lambda_{k} v^{2}(\pi - \theta)}\cdot e^{\lambda_{k}\pi \left(\frac{va(\theta)}{2\pi} \right)^{2}} \int_{t=\sqrt{2\pi\lambda_{k}}\frac{va(\theta)}{2\pi}}^{\infty} e^{-t^{2}/2}\left(t - va(\theta)\sqrt{\frac{\lambda_{k}}{2\pi}}\right)\mathrm{d}t\mathrm{d}\theta \\
= 1 - \frac{1}{\pi}\int_{\theta = 0}^{\pi}  \left [ 1 - 2 b(v, \lambda_{k}, \theta)\sqrt{\pi}e^{b^{2}(v, \lambda_{k}, \theta)} Q(\sqrt{2}b(v, \lambda_{k}, \theta)) \right] \exp \left( -\lambda_{k} v^{2} (\pi - \theta) \right) \mathrm{d}\theta.
\end{dmath}
Setting $b(v, \lambda_{k}, \theta) = \sqrt{\pi \lambda_{k}} \frac{v a(\theta)}{2\pi}$ and employing the change of variable $t = \sqrt{2\pi\lambda_{k}}(r + \frac{va(\theta)}{2 \pi})$ gives the desired result, and the proof is complete
\end{proof}

In Fig.~\ref{fig:handoff}, we examine the validity of the expressions for the handoff rate. As seen in the figure, since~\eqref{eq:specialcase} is an exact expression, the numerical simulations exactly match the analysis for the special case where the user is moving radially away from the serving AP. As expected intuitively, the handoff rate increases with the increase in the user displacement or the AP density. For the general case, the plots are obtained using the approximate expression in~\eqref{eq:handoffRateGen}, and the integral is obtained numerically. The slight deviation of the analysis from the numerical simulations for large $v$ in Fig.~\ref{fig:VsVel} or very large AP density in Fig.~\ref{fig:VsDen} is the result of user displacement becoming comparable to the connection distance which is inconsistent with the approximation. However, for reasonable speeds and AP densities, the approximation in~\eqref{eq:handoffRateGen} is quite accurate and will be used throughout the paper.

The handoff rate, derived in this section, allows us to incorporate the user mobility in the coverage analysis and tier association discussed in the following sections.
\begin{figure}
     \begin{center}
        \subfigure[Handoff rate versus user displacement. $\lambda_{k} = 1/(1000 m^{2})$.]{%
            \label{fig:VsVel}
            \includegraphics[width=0.5\textwidth]{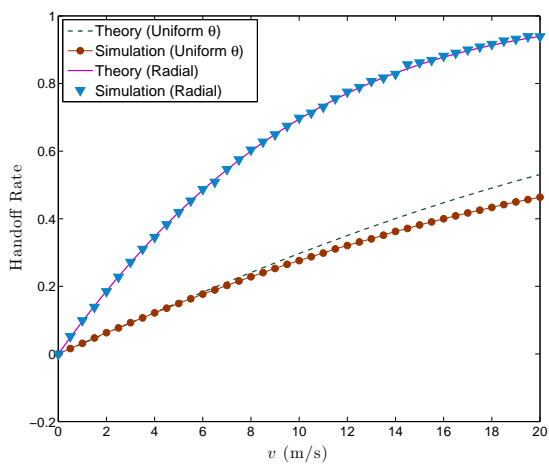}
        }%
        \subfigure[Handoff rate versus AP density. $v = 5$m/s.]{%
           \label{fig:VsDen}
           \includegraphics[width=0.5\textwidth]{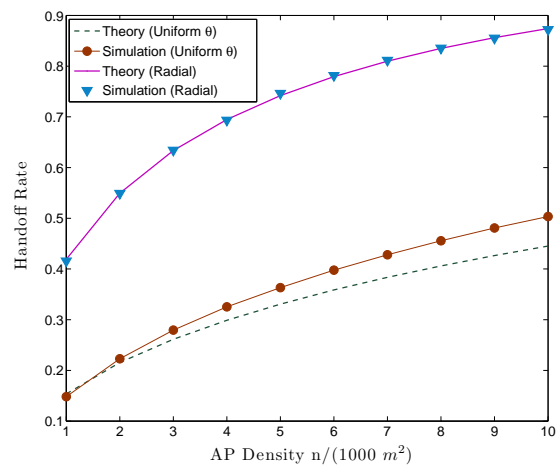}
        } 
        \end{center}
    \caption{%
       Handoff rate versus: (a) user displacement in a unit of time $v$, (b) AP density $\lambda_{k}$, for both the general case ($\theta$ has uniform distribution) and the special case of radial movement ($\theta = 0$).
 }%
   \label{fig:handoff}
\end{figure}

\section{Coverage Probability with Handoffs}
\label{coverage1}
In~\cite{andrews:11}, the probability of coverage in an interference-limited, single-tier network with the AP locations modelled by a homogeneous PPP, was derived to be:
\begin{equation}
\label{CovAndrews}
\mathds{P}(\gamma \geq \tau) = \frac{1}{1 + \rho(\tau, \alpha)},
\end{equation}
where $\gamma$ is the received SIR at the user, and $\tau$ denotes the SIR threshold. For the case where both the desired and the interfering signals undergo Rayleigh fading, $\rho(\tau, \alpha)$ is given by:
\begin{equation}
\rho(\tau, \alpha) = \tau^{2/\alpha}\int_{\tau^{-2/ \alpha}}^{\infty}\frac{1}{1 + u^{\alpha/2}}\mathrm{d}u.
\end{equation}
The expression given in~\eqref{CovAndrews} is the probability that a single user is in coverage; it can also be interpreted as the fraction of all users in coverage at any given time. Each user connects to the AP with the largest signal power. In the analysis of~\cite{andrews:11}, the users are stationary, and there are no handoffs. Furthermore, with the focus on SIR,~\eqref{CovAndrews} indicates that the probability of coverage is independent of the AP transmit power or density. In other words, increasing the density of APs or transmit power will not affect the resulting SIR. Consequently, when handoffs are not accounted for, the network capacity increases linearly with the number of APs. However, for a mobile user, as shown in Fig.~\ref{fig:VsDen}, the handoff rate increases with the AP density which negatively affects the system performance.

The cost associated with handoffs is due to service delays or dropped calls. The higher the handoff rate, the higher the chance of degradation in the quality of service. Once we incorporate mobility, a fraction of users that initially met the SIR criterion for coverage would experience a connection failure due to the handoff. This fraction is determined by the system sensitivity to handoffs. Conversely, if the user is in coverage and no handoff occurs, it stays in coverage and continues to receive service. To take user mobility into account, we consider a linear function that reflects the cost of handoffs. Under this model, the probability of coverage is given by:
\begin{equation}
\label{CovCost}
P_{c}(v, \lambda_{k}, \beta,\tau_{k}, \alpha) = \mathds{P}\left(\gamma_{k} \geq \tau_{k}, \widebar{H}_{k} \right) + (1 - \beta) \mathds{P}\left(\gamma_{k} \geq \tau_{k}, H_{k}\right),
\end{equation}
where $\gamma_{k}$ is the received SIR at the user. The first term is the probability of the joint event that the user is in coverage and no handoff occurs. The second term is the probability of the joint event that the user is in coverage and handoff occurs penalized by the cost of handoff; here, $\beta \in [0,1]$ is the probability of connection failure due to the handoff, i.e., a fraction $\beta$ of handoffs result in dropped connections even though the user is in coverage from the SIR point-of-view. 
The coefficient $\beta$, in effect, measures the system sensitivity to handoffs. Its value depends on a number of factors, e.g., the radio access technology, the mobility protocol, the protocol's layer of operation and the link speed~\cite{banerjee:03,nakajima:03,Lin:07}. At one extreme, as $\beta \rightarrow 1$, $P_{c}(v, \lambda_{k}, \beta,\tau_{k}, \alpha) = \mathds{P}\left(\gamma_{k} \geq \tau_{k}, \widebar{H}_{k} \right)$, stating that only those users that are initially in coverage and do not undergo a handoff maintain their connection, i.e., every handoff results in an outage. At the other extreme, as $\beta \rightarrow 0$, the system is not sensitive to handoffs and the expression for the probability of coverage reduces to~\eqref{CovAndrews}, since $\mathds{P}\left(\gamma_{k} \geq \tau_{k}, \widebar{H}_{k} \right) + \mathds{P}\left(\gamma_{k} \geq \tau_{k}, H_{k}\right) = \mathds{P}(\gamma_{k} \geq \tau_{k})$. It is worth noting that~\eqref{CovCost} predicts the coverage probability for a mobile user, already in coverage, immediately after knowing the user's speed and the system sensitivity to handoffs. 

Using the handoff rate derived in Section~\ref{sec:hoffrate} and the overall probability of coverage given in~\eqref{CovCost}, we can incorporate the user mobility in the coverage analysis as follows.
\begin{thm}
\label{propos:CovHandoff}
The probability of coverage $P_{c}(v, \lambda_{k}, \beta,\tau_{k}, \alpha)$ for a typical mobile user moving a distance $v$ in a unit of time in a network with access points distributed according to a homogeneous PPP with density $\lambda_{k}$ is given by:
\begin{multline}
P_{c}(v,\lambda_{k}, \beta, \tau_{k}, \alpha) = \frac{1}{1 + \rho_{k}} \cdot\\
\left \{(1 - \beta) + \beta \left[ \frac{1}{\pi}\int_{\theta = 0}^{\pi}  \left[ 1 -  2 b' \sqrt{\pi} e^{{b'}^2} Q(\sqrt{2}b') \right] \exp\Big(-\lambda_{k} v^2(\pi - \theta)\Big) \mathrm{d}\theta \right] \right\}
\end{multline}
where $b' = b'(v, \lambda_{k}, \theta, \tau_{k}, \alpha) = \frac{va(\theta)}{2\pi}\sqrt{\frac{\pi\lambda_{k}}{1 + \rho_{k}}}$ and $\rho_{k} = \rho(\tau_{k}, \alpha)$.
\end{thm}
\begin{proof}
From~\eqref{CovCost}, the probability of coverage conditioned on $r$ and $\theta$ is given by:
\begin{equation}
\mathds{P}\left(\gamma_{k} \geq \tau_{k}, \widebar{H}_{k} | r, \theta \right) + (1 - \beta)\mathds{P}\left(\gamma_{k} \geq \tau_{k},H_{k} | r, \theta \right),
\end{equation}
where
\begin{equation}
\begin{array}{ll}
\mathds{P}\left(\gamma_{k} \geq \tau_{k},H_{k} | r, \theta \right) & = \mathds{P}(\gamma_{k} \geq \tau_{k}|r)\cdot \mathds{P}(H_{k} | r, \theta) \\
& = \mathds{P}(\gamma_{k} \geq \tau_{k}|r)\cdot \Big(1 - \mathds{P}(\widebar{H}_{k} | r, \theta)  \Big) \\
& = \mathds{P}(\gamma_{k} \geq \tau_{k}|r) - \mathds{P}(\gamma_{k} \geq \tau_{k}|r)\cdot \mathds{P}(\widebar{H}_{k} |r,\theta)\\
& = \mathds{P}(\gamma_{k} \geq \tau_{k}|r) - \mathds{P}(\gamma_{k} \geq \tau_{k}, \widebar{H}_{k} |r,\theta).
\end{array}
\end{equation}
Hence,
\begin{equation}
\begin{array}{l}
\label{CovNoHandoff}
\mathds{P}\left(\gamma_{k} \geq \tau_{k}, \widebar{H}_{k} | r, \theta \right) + (1 - \beta)\mathds{P}\left(\gamma_{k} \geq \tau_{k},H_{k} | r, \theta \right) \\
  = \mathds{P}\left(\gamma_{k} \geq \tau_{k}, \widebar{H}_{k} | r, \theta \right) + (1 - \beta)\Big( \mathds{P}(\gamma_{k} \geq \tau_{k}|r) - \mathds{P}(\gamma_{k} \geq \tau_{k}, \widebar{H}_{k} |r,\theta) \Big) \\
 = (1 - \beta)\mathds{P}(\gamma_{k} \geq \tau_{k}|r) + \beta \mathds{P}(\gamma_{k} \geq \tau_{k}, \widebar{H}_{k} |r,\theta).
\end{array}
\end{equation}
The probability of coverage with Rayleigh fading at connection distance $r$ is given by~\cite[Theorem 2]{andrews:11}:
\begin{equation}
\label{eq1}
\mathds{P}\left( \gamma_{k} \geq \tau_{k}|r \right) = e^{-\pi \lambda_{k} r^{2}\rho(\tau_{k}, \alpha)},
\end{equation}
therefore:
\begin{equation}
\label{eq:Final}
\begin{array}{ll}
P_{c}(v,\lambda_{k}, \beta, \tau_{k}, \alpha) & = \mathds{E}_{\Theta}\Big[\mathds{E}_{R_{k}}\Big[\mathds{P}( \gamma_{k} \geq \tau_{k} |r, \theta) \Big] \Big] \\
& = (1 - \beta) \mathds{P}(\gamma_{k} \geq \tau_{k}) +  \beta \mathds{E}_{\Theta}\Big[\mathds{E}_{R_{k}}\Big[\mathds{P}( \gamma_{k} \geq \tau_{k}, \widebar{H}_{k} |r, \theta) \Big] \Big] \\
& = \frac{1 - \beta}{1 + \rho(\tau_{k}, \alpha)} + \beta \mathds{E}_{\Theta}\Big[\mathds{E}_{R_{k}}\Big[\mathds{P}( \gamma_{k} \geq \tau_{k}, \widebar{H}_{k} |r, \theta) \Big] \Big].
\end{array}
\end{equation}
In the second term, the probability of the joint event that the user is in coverage and the handoff does not occur is given by:
\begin{equation}
\begin{split}
\label{term2}
\mathds{P}\left( \gamma_{k} \geq \tau_{k}, \widebar{H}_{k} \right) & = \mathds{E}_{\Theta}\Big[\mathds{E}_{R_{k}}\left[\mathds{P}\left( \gamma_{k} \geq \tau_{k}, \widebar{H}_{k}|r, \theta \right)  \Big] \right] \\
& = \displaystyle \frac{1}{\pi} \int_{\theta= 0}^{\pi}\int_{r=0}^{\infty}\mathds{P}\left( \gamma_{k} \geq \tau_{k}|r \right)\cdot \mathds{P} \left(\widebar{H}_{k}|r, \theta \right)\cdot f_{R_{k}}\left(r\right)  \mathrm{d}r \mathrm{d}\theta \\
\displaystyle & \qeq \frac{1}{\pi} \int_{\theta= 0}^{\pi}\int_{r=0}^{\infty} e^{-\pi \lambda_{k} r^{2}\rho(\tau_{k}, \alpha)} \cdot e^{-\lambda_{k}\left(v^{2}(\pi - \theta) + rva(\theta) \right)}\cdot 2\pi \lambda_{k} r  e ^{- \lambda_{k} \pi r^{2}}  \mathrm{d}r \mathrm{d}\theta \\
\displaystyle & = \frac{1}{\pi} \displaystyle \int_{\theta= 0}^{\pi} e^{-\lambda_{k} v^{2}(\pi - \theta)} \int_{r= 0}^{\infty} e^{-\lambda_{k} \pi (1 + \rho_{k}) \Big(\left(r + \frac{va(\theta)}{2\pi (1 + \rho_{k})} \right)^{2} - (\frac{va(\theta)}{2\pi(1 + \rho_{k})})^{2} \Big)}2\pi \lambda_{k} r \mathrm{d}r \mathrm{d}\theta \\
& = \frac{1}{\pi} \left(\frac{1}{1 + \rho_{k}}\right) \int_{\theta = 0}^{\pi} e^{-\lambda_{k} v^2(\pi - \theta)} \left[ 1 -  2b' \sqrt{\pi} e^{{b'}^2} Q(\sqrt{2}b') \right]  \mathrm{d}\theta, \\
\end{split}
\end{equation}
where $\rho_{k} = \rho(\tau_{k}, \alpha)$ and $(a)$ follows from using the probability of the complementary event in~\eqref{eq:handoffApprox} given by:
\begin{equation}
\begin{array}{ll}
\mathds{P} \left(\widebar{H}_{k}|r, \theta \right) & = 1 - \mathds{P} \left(H_{k}|r, \theta \right) \\
& = \exp \Bigg( -\lambda_{k} \Big[v^{2}(\pi - \theta) +  rv a(\theta) \Big] \Bigg).
\end{array}
\end{equation}
The final two steps are similar to the proof for Corollary~\ref{propos:handoffRateGen}, employing the change of variable $t=\sqrt{2\pi\lambda_{k}(1 + \rho_{k})}\left(r + \frac{va(\theta)}{2 \pi (1 + \rho_{k})}\right)$ and setting $b'(v, \lambda_{k}, \theta, \tau_{k}, \alpha) = \frac{va(\theta)}{2\pi}\sqrt{\frac{\pi\lambda_{k}}{1 + \rho_{k}}}$. Finally, using~\eqref{term2} in~\eqref{eq:Final} gives the desired result, and the proof is complete.
\end{proof}

The simulation results in Fig.~\ref{fig:CovHandoffEffect} show the effect of mobility on coverage. In both figures, the probability of coverage for stationary users is $1/(1 + \rho(\tau_{k},\alpha)) = 0.49$ for $\tau_{k} = 0$dB, and $\alpha = 3.5$. For mobile users, on the other hand, the probability of coverage decreases with the increase in the handoff rate; as expected, this negative effect is more noticeable when the probability of connection failure due to handoffs is significant (i.e., large $\beta$ as in Fig.~\ref{fig:CovVsVelSmallKappa}). While the handoff rate increases linearly with the user speed in a network with low AP density, it saturates in a high density network.

The plots in Fig.~\ref{fig:CovVsTau} show the probability of coverage versus the SIR threshold $\tau_{k}$ for a mobile user for different AP densities. As expected, although the SIR distribution remains the same regardless of the AP density, the probability of coverage in a network with a higher AP density decreases due to frequent handoffs. The degradation in coverage not only depends on the network sensitivity to handoffs, determined by $\beta$, but also on the SIR threshold. While the probability of coverage in a network with a high AP density is lower than that with a low AP density, the difference between the two is more noticeable at lower SIR thresholds, or in a network with a large $\beta$. It is this interplay between mobility and tier association that leads us to consider mobility-aware tier association in the next section.
\begin{figure}
     \begin{center}
        \subfigure[Probability of coverage versus $v$; $\beta = 0.3$.]{
           \label{fig:CovVsVelLargeKappa}
           \includegraphics[width = 0.47\textwidth,height = 0.38\textwidth]{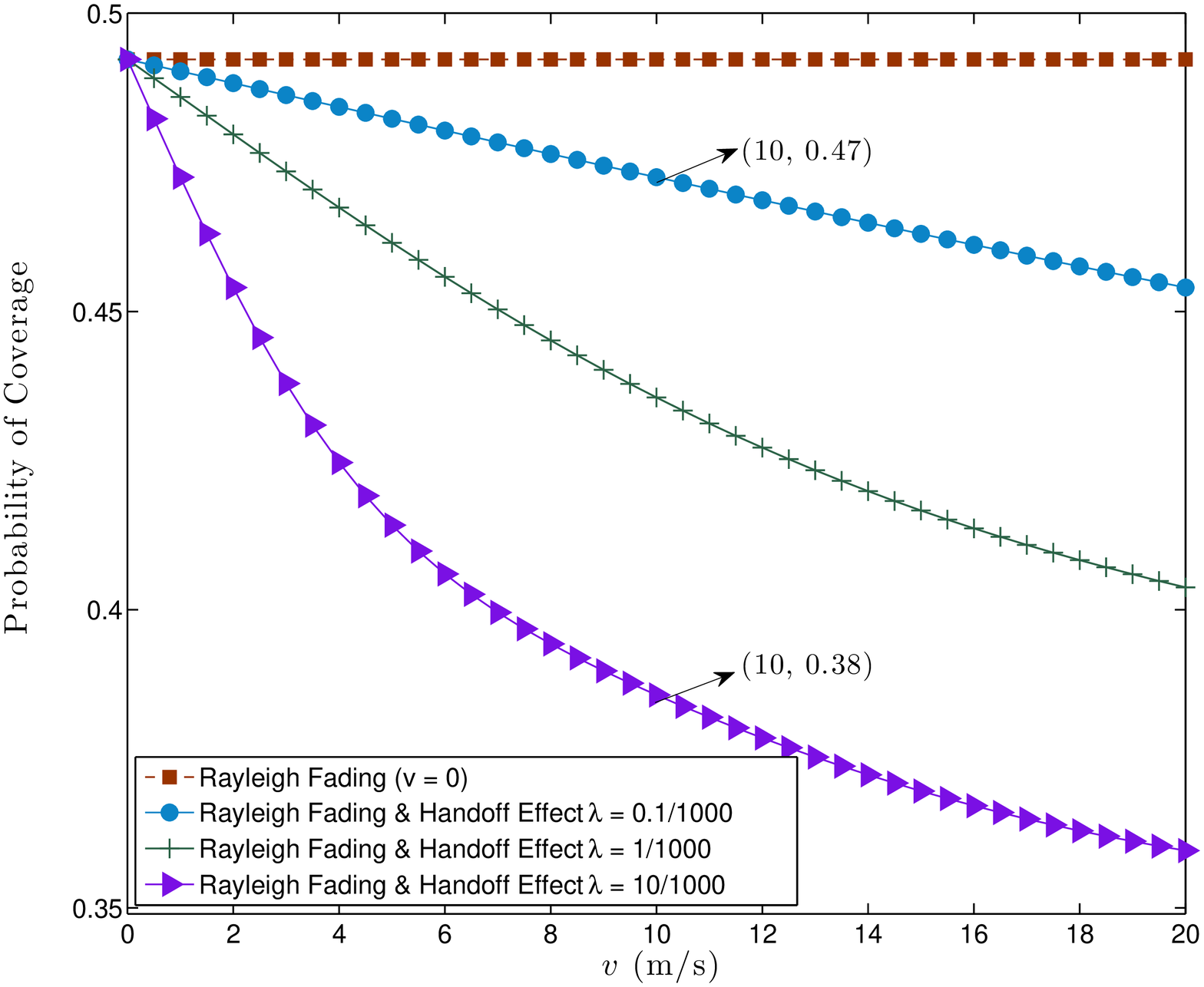}
        }
        \subfigure[Probability of coverage versus $v$; $\beta = 0.9$.]{
           \label{fig:CovVsVelSmallKappa}
           \includegraphics[width=0.47\textwidth,height = 0.38\textwidth]{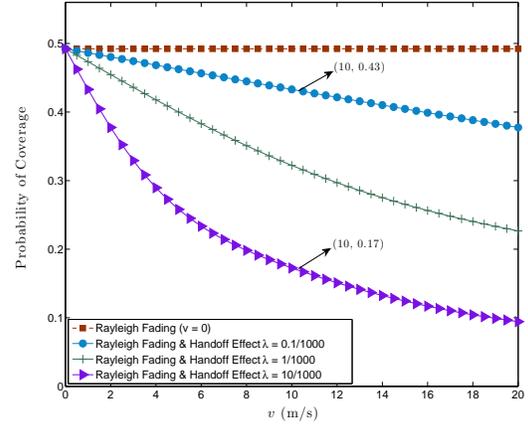}
        }
        \end{center}
    \caption{%
        Probability of coverage versus user displacement $v$ in a unit of time for different AP densities and $\tau_{k} = 0$dB: (a) the system is less sensitive to handoffs, $\beta = 0.3$; (b) the probability of connection failure due to handoffs is large, $\beta = 0.9$.}
   \label{fig:CovHandoffEffect}
\end{figure}

 \begin{figure}
     \begin{center}
        \subfigure[Probability of coverage versus $\tau_{k}$; $\beta = 0.3$.]{%
            \label{fig:CovVsKappa}
            \includegraphics[width=0.47\textwidth]{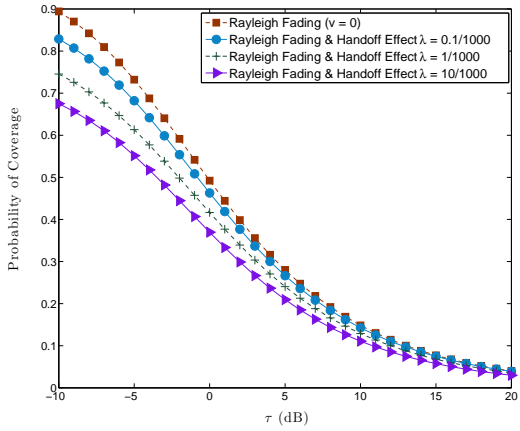}
        }%
        \subfigure[Probability of coverage versus $\tau_{k}$; $\beta = 0.9$.]{%
           \includegraphics[width=0.47\textwidth]{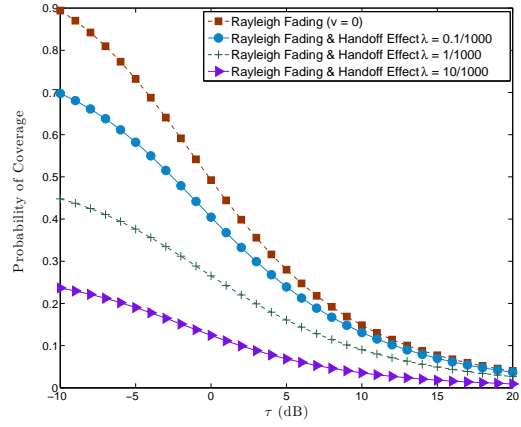}

        }
        \end{center}
    \caption{%
        Probability of coverage versus SIR threshold $\tau_{k}$ for: (a) $\beta = 0.3$ , (b) $\beta = 0.9$. $v = 15$ in both figures.}%
   \label{fig:CovVsTau}
\end{figure}


\section{Mobility-Aware Tier Association}
\label{coverage2}
In the previous section, we derived the probability of coverage for a mobile user in a single-tier network, and showed how it is affected by the handoff rate. The dependence of the handoff rate on the AP density is the rationale for associating fast moving users with the higher tiers (with smaller AP densities) to compensate for the potential connection failure due to handoffs. In practice, this would be achieved by adjusting the tier's bias factor. In this section, we first present the probability of coverage and the optimum tier association and bias factor for a \textit{stationary} user in a network. Then, as in Section~\ref{coverage1}, we incorporate mobility in the coverage analysis assuming a linear cost function to reflect the impact of handoffs. Finally, we optimize tier association in a multi-tier network to account for the effect of handoffs in maximizing the overall probability of coverage for a mobile user .

In a multi-tier network with the maximum biased average received power as the tier association metric, the probability that a user connects to tier $k$ is determined by the tier's AP density, $\lambda_k$, transmit power, $P_k$, and bias factor, $B_k$, and has been shown to be~\cite{shin:12}:
 \begin{equation}
\label{layerProb}
A_{k} = \frac{\lambda_{k}(P_{k} B_{k})^{2/\alpha}}{\sum_{j = 1}^{K} \lambda_{j} (P_{j} B_{j})^{2/\alpha}} = \frac{1}{\sum_{j = 1}^{K} \widehat{\lambda}_{j} (\widehat{P}_{j}\widehat{B}_{j})^{2/\alpha}},
\end{equation}
where $\widehat{P}_{j} = \frac{P_{j}}{P_{k}}, \widehat{B}_{j} = \frac{B_{j}}{B_{k}}, \widehat{\lambda}_{j} = \frac{\lambda_{j}}{\lambda_{k}}$. While the transmit power and the AP density are mostly determined by the network infrastructure, adjusting the bias factor can dynamically change the user association to different tiers in the network. Let $n$ denote the index of the tier associated to the user. Since the user connects to only one tier at a time, the probability that the user connects to tier $k$ at connection distance $r$ is given by~\cite[Lemma 1]{shin:12} $\mathds{P}(n = k|r) = \prod_{j = 1,j\neq k}^{K}e^{-\pi \lambda_{j}(\widehat{P}_{j}\widehat{B}_{j})^{2/\alpha}r^{2}}.$ The probability of coverage for a stationary user in a multi-tier network with the spectrum shared across the network is derived in~\cite{shin:12}. Our focus is on the special case of \textit{orthogonal} spectrum allocation amongst tiers and the corresponding optimum tier association. The optimum tier association and the bias factor for the maximum SIR coverage for a single-tier two-RAT\footnote{RAT refers to radio access technology. Each RAT is allocated a different frequency of operation.} network was derived in~\cite[Proposition 1]{Singh:13}. Generalizing this result to a $K$-tier network with orthogonal spectrum allocation across tiers is straight forward and we present it below for use later:
\begin{propos}
\label{Proposition1}
\begin{itemize}
\item[(a)] The probability of coverage for a randomly located user in a multi-tier network with orthogonal\footnote{Note that in a multi-tier network with all tiers sharing the same spectrum, the set of interfering APs include all the APs in the network except the serving AP in the serving tier. In this setup with the same tier association metric, the overall probability of coverage is given by~\eqref{OverallCovThrm} with the term $\rho(\tau_{k}, \alpha)$ replaced by $\sum_{j=1}^{K} \widehat{\lambda}_{j}\widehat{P}_{j}^{2/\alpha} \mathcal{Z}(\tau_{k}, \alpha, \widehat{B}_{j})$, where $\mathcal{Z}(\tau_{k}, \alpha, \widehat{B}_{j})$ is given by~\cite{shin:12}:
\begin{equation*}
\mathcal{Z}(\tau_{k}, \alpha, \widehat{B}_{j}) = \tau_{k}^{2/\alpha}\int_{(\widehat{B}_{j}/\tau_{k})^{2/\alpha}}^{\infty}\frac{1}{1 + u^{\alpha/2}}\mathrm {d}u
\end{equation*}
to account for the interference from the other tiers. For such a case, this expression would replace the corresponding expression in the analysis for orthogonal allocation given below. However, it appears that no closed-form expressions are possible and one must resort to numerical solutions.} spectrum allocation among tiers, and the APs in each tier distributed according to a homogeneous PPP with density $\lambda_{k}$, is given by:
\begin{equation}
\label{OverallCovThrm}
\boldsymbol{P}_{c} = \sum_{k=1}^{K}\frac{1}{A_{k}^{-1} + \rho(\tau_{k}, \alpha)}.
\end{equation}
where $A_{k}$ is the tier association probability.
\item[(b)] The optimum tier association probability to maximize the SIR coverage is then given by:
\begin{equation}
\label{partA}
A_{k}^{*} = \displaystyle \frac{1/\rho(\tau_{k}, \alpha)}{\sum_{k = 1}^{K}1/\rho(\tau_{k}, \alpha)}.
\end{equation}
\end{itemize}
\end{propos}
\begin{proof}
The proof of Proposition is omitted as it is a straightforward extension of the proof in~\cite[Proposition 1]{Singh:13}.
\end{proof}
It is easy to see that when all tiers have the same SIR threshold, i.e., $\tau_{k} = \tau \; \; \forall k$, the maximum probability of coverage is achieved with equal tier association, i.e., $A_{k}^{*} = 1/K$, and the coverage probability is given by:
\begin{equation}
\boldsymbol{P}^{*}_{c} = \frac{K}{K + \rho(\tau, \alpha)}.
\end{equation}

Importantly, given the tier association probabilities, the optimum bias factors can be found uniquely by solving a system of linear equations.

Define $\textbf{x}^{*} = [x_{2}, x_{3}, \dots, x_{K}]^\top$ with $x_{k} = B_{k}^{2/\alpha}$, $k = 2, 3, \cdots K$ and $x_{1} = B_{1} = 1$ for the uppermost tier. $[~\cdot~]^\top$ denotes the transpose operation. From~\eqref{layerProb}, the tier association probability for tier $k$ can be rewritten as:
\begin{equation}
\begin{array}{ll}
\label{LinearEq}
A_{k}^{-1} & = \displaystyle \sum_{j = 1}^{K}\widehat{\lambda_{j}}\widehat{P_{j}}^{2/\alpha}B_{j}^{2/\alpha}\cdot B_{k}^{-2/\alpha} \\
& = 1 + \displaystyle \sum_{j = 1, j\neq k}^{K}\widehat{\lambda_{j}}\widehat{P_{j}}^{2/\alpha}B_{j}^{2/\alpha}\cdot B_{k}^{-2/\alpha}. \\
\Rightarrow & \left(1 - A_{k}^{-1}\right)B_{k}^{2/\alpha} + \displaystyle \sum_{j = 1, j\neq k}^{K}\left(\widehat{\lambda_{j}}\widehat{P_{j}}^{2/\alpha}\right)B_{j}^{2/\alpha} = 0.
\end{array}
\end{equation}
Setting $a_{jk} = \widehat{\lambda_{j}}\widehat{P_{j}}^{2/\alpha}$, the last line in~\eqref{LinearEq} can be written as:
\begin{equation}
\left(1 - A_{k}^{-1}\right)x_{k} + \displaystyle \sum_{j = 2, j\neq k}^{K}{a}_{jk}x_{j} = -{a}_{1k} \; \; \; k = 2, 3, \cdots K.
\end{equation}
Given the optimum association probabilities, $\{A_{k}^{*}\}_{k = 1}^{K}$, the optimum vector $\textbf{x}^{*} = [x_{2}, x_{3}, \dots, x_{K}]^\top$ is the unique solution to $\textbf{A}\textbf{x} = \textbf{b}$ where $\textbf{b} = [-{a}_{12},-{a}_{13}, \dots, -{a}_{1K}]^\top$, and $\textbf{A}$ is given by:
\begin{equation}
\textbf{A} = \begin{bmatrix}
	(1 - {A_{2}^{*}}^{-1}) 	& a_{32}		& \dots 	& a_{K2} 		\nonumber\\
	a_{23} 	& (1 - {A_{3}^{*}}^{-1})	& \dots 	& a_{K3} 		\nonumber\\
	\vdots 		& \vdots 			& \vdots 	& \vdots			\nonumber\\
	a_{2K} 	& a_{3K}		& \dots 	& (1 - {A_{K}^{*}}^{-1})	\nonumber\\
	\end{bmatrix}. \nonumber
\end{equation}
Note that this result is not limited to the case of optimum tier association probabilities. Given any set of non-zero tier association probabilities, the corresponding bias factors can be found solving the system of linear equations above. The following theorem states that the matrix above is full-rank and, hence, the relationship between the association probabilities and bias factors is one-to-one.
\begin{thm}
\label{FullRank}
The matrix $\mathbf{A}$ is full-rank.
\end{thm}
\begin{proof}
See Appendix \ref{AppendixB}.
\end{proof}

The expressions for the probability of coverage and the tier association provided above do not take into account user mobility, the associated handoffs and the connection failure due to such handoffs. Using the handoff rate derived in Section~\ref{sec:hoffrate} and the linear cost function given in~\eqref{CovCost}, we can generalize the results in Section~\ref{coverage1} to incorporate user mobility in the probability of coverage in a multi-tier network as follows:
\begin{thm}
\label{propos:CovHandoffMultLay}
The probability of coverage $P_{c}(v, \{\lambda_{k}\}_{k = 1}^{K}, \beta, \{\tau_{k}\}_{k = 1}^{K}, \alpha)$ for a typical mobile user moving a distance $v$ in a unit of time in a multi-tier network with the biased average received power as the tier connection metric and the access points of tier $k$ distributed according to a homogeneous PPP with density $\lambda_{k}$ is given by:
\begin{multline}
\label{eq:MultLay}
P_{c}(v, \{\lambda_{k}\}_{k = 1}^{K}, \beta, \{\tau_{k}\}_{k = 1}^{K}, \alpha) =  \\
\sum_{k=1}^{K}\frac{1}{A_{k}^{-1} + \rho(\tau_{k}, \alpha)} \left \{(1 - \beta) + \beta \left[ \frac{1}{\pi} \int_{\theta = 0}^{\pi} \left[ 1 - 2 b''_{k}\sqrt{\pi} e^{{b_{k}^{\second}}^{2}}  Q(\sqrt{2}b''_{k}) \right] \exp \Big(-\lambda_{k})v^2(\pi - \theta) \Big)
\mathrm{d}\theta \right] \right \}
\end{multline}
where $b''_{k} = b''_{k} (v, \lambda_{k}, \theta, \tau_{k}, \alpha, A_{k}) =  \frac{va(\theta)}{2\pi}\sqrt{\frac{\pi\lambda_{k}}{A_{k}^{-1} + \rho(\tau_{k}, \alpha)}}$ and $a(\theta) = 2\cos \theta (\pi - \theta) + \sin \theta$.
\end{thm}
\begin{proof}
Specializing~\eqref{CovNoHandoff} to tier $k$, we have:
\begin{equation}
\begin{array}{l}
\mathds{P}\left(\gamma_{k} \geq \tau_{k}, n = k, \widebar{H}_{k} | r, \theta \right) + (1 - \beta) \mathds{P}\left(\gamma_{k} \geq \tau_{k}, n = k, H_{k} | r, \theta \right) \\
= (1 - \beta)\mathds{P}(\gamma_{k} \geq \tau_{k}, n = k|r) + \beta \mathds{P}(\gamma_{k} \geq \tau_{k}, n = k, \widebar{H}_{k} |r,\theta).
\end{array}
\end{equation}
The probability of the joint event that the user connects to tier $k$, is in coverage and a handoff does not occur is given by:
\begin{equation}
\begin{array}{l}
\label{eq:Term2}
\mathds{P}\left( \gamma_{k} \geq \tau_{k}, n = k, \widebar{H}_{k} \right) = \mathds{E}_{\Theta}\Big[\mathds{E}_{R_{k}}\left[\mathds{P}\left( \gamma_{k} \geq \tau_{k},n = k, \widebar{H}_{k}|r, \theta \right)  \Big] \right] \vspace{0.1in} \\
= \displaystyle \frac{1}{\pi} \int_{\theta= 0}^{\pi}\int_{r=0}^{\infty}\mathds{P}\left( \gamma_{k} \geq \tau_{k}|r \right)\cdot \mathds{P}(n = k|r) \cdot \mathds{P} \left(\widebar{H}_{k}|r, \theta \right)\cdot f_{R_{k}}\left(r\right)  \mathrm{d}r \mathrm{d}\theta \\
 = \displaystyle \frac{1}{\pi} \int_{\theta= 0}^{\pi}\int_{r=0}^{\infty} e^{-\pi \lambda_{k} r^{2}\rho(\tau_{k}, \alpha)} \cdot \left(\prod_{j = 1,j\neq k}^{K}e^{-\pi \lambda_{j}(\widehat{P}_{j}\widehat{B}_{j})^{2/\alpha}r^{2}}\right) \cdot e^{-\lambda_{k}\left(v^{2}(\pi - \theta) + rva(\theta) \right)}\cdot 2\pi \lambda_{k} r  e ^{- \lambda_{k} \pi r^{2}}  \mathrm{d}r \mathrm{d}\theta \\
 \displaystyle  = \frac{1}{\pi} \displaystyle \int_{\theta= 0}^{\pi} e^{-\lambda_{k} v^{2}(\pi - \theta)} \int_{r= 0}^{\infty}e^{-\pi \lambda_{k} r^{2} \Big[  \rho(\tau_{k}, \alpha) + \sum_{j=1}^{K} \widehat{\lambda}_{j}(\widehat{P}_{j}\widehat{B}_{j})^{2/\alpha}\Big]}\cdot  2\pi \lambda_{k} e^{-\lambda_{k}rva(\theta)} \mathrm{d}r \mathrm{d}\theta\\
\displaystyle  = \frac{1}{\pi} \displaystyle \int_{\theta= 0}^{\pi} e^{-\lambda_{k} v^{2}(\pi - \theta)} \int_{r= 0}^{\infty} e^{-\lambda_{k} \pi (A_{k}^{-1} + \rho_{k}) \Big(\left(r + \frac{va(\theta)}{2\pi (A_{k}^{-1} + \rho_{k})} \right)^{2} - (\frac{va(\theta)}{2\pi(A_{k}^{-1} + \rho_{k})})^{2} \Big)}2\pi \lambda_{k} r \mathrm{d}r \mathrm{d}\theta \\
\displaystyle = \frac{1}{\pi}\left(\frac{1}{A_{k}^{-1} + \rho_{k}}\right) \int_{\theta = 0}^{\pi} e^{-\lambda_{k} v^2(\pi - \theta)} \left[ 1 -  2b_{k}'' \sqrt{\pi} e^{{b^{\second}}^2} Q(\sqrt{2}b_{k}'') \right]  \mathrm{d}\theta,
\end{array}
\end{equation}
where $\rho_{k} = \rho(\tau_{k}, \alpha)$. The change of variable $t = \sqrt{2\pi\lambda_{k}(A_{k}^{-1} + \rho_{k})} (r + \frac{va(\theta)}{2\pi(A^{-1}_{k} + \rho_{k})})$, and setting $b_{k}''(v, \lambda_{k}, \theta, \tau_{k}, \alpha, A_{k}) = \frac{va(\theta)}{2\pi} \sqrt{\frac{\pi \lambda_{k}}{A^{-1}_{k} + \rho_{k}}}$ gives the final expression. Using the sum probability of disjoint events, the probability of coverage in a multi-tier network is then given by:
\begin{equation}
\begin{array}{l}
\label{Final}
P_{c}(v, \{\lambda_{k}\}_{k = 1}^{K}, \beta, \{\tau_{k}\}_{k = 1}^{K}, \alpha)  \displaystyle = \sum_{k=1}^{K} \bigg[(1 - \beta)\mathds{P}(\gamma_{k} \geq \tau_{k}, n = k) + \beta \mathds{P}(\gamma_{k} \geq \tau_{k}, n = k, \widebar{H}_{k}) \bigg] \\
 =  \displaystyle \sum_{k = 1}^{K} \bigg[ \left(\frac{1 - \beta}{A_{k}^{-1} + \rho_{k}}\right) + \frac{\beta}{\pi}\left(\frac{1}{A_{k}^{-1} + \rho_{k}}\right) \int_{\theta = 0}^{\pi} e^{-\lambda_{k} v^2(\pi - \theta)} \left[ 1 -  2b_{k}'' \sqrt{\pi} e^{{b_{k}^{\second}}^2} Q(\sqrt{2}b_{k}'') \right]  \mathrm{d}\theta\bigg] \\
 =  \displaystyle \sum_{k=1}^{K}\frac{1}{A_{k}^{-1} + \rho_{k}}
\left \{(1 - \beta) + \beta \left[ \frac{1}{\pi} \int_{\theta = 0}^{\pi} \left[ 1 - 2 b''_{k} \sqrt{\pi} e^{{b_{k}^{\second}}^{2}}  Q(\sqrt{2}b''_{k}) \right] \exp \Big(-\lambda_{k}v^2(\pi - \theta) \Big)
\mathrm{d}\theta \right] \right \},
\end{array}
\end{equation}
and the proof is complete.
\end{proof}
As in the single-tier case, if the user is stationary ($v = 0$), or there is no connection failure due to the handoff ($\beta = 0$), the expression for the overall probability of coverage reduces to the expression in Proposition~\ref{Proposition1}.

To maximize the overall probability of coverage, let the tier association probabilities $\{A_{i} \}_{i=1}^{K}$ become the set of optimization variables. The optimization problem with the objective of maximizing the overall probability of coverage can then be formulated as:
\begin{equation}
\begin{array} {ll}
\label{OptProb}
P_c^* = \displaystyle \max_{\{A_{i}\}_{i =1}^{K}} & \displaystyle \sum_{i=1}^{K}\Big[(1 - \beta)f_{i,1}(A_{i}) + \beta f_{i,2}(A_{i})\Big] \\
\mbox{subject to:}  & \displaystyle \sum_{i = 1}^{K} A_{i} = 1 \\
& A_{i} \geq 0 \; \; i = 1,\cdots K,
\end{array}
\end{equation}
where the objective function in~\eqref{OptProb} is the same expression as in~\eqref{eq:MultLay}. Here, $f_{i,1} (A_{i})= \frac{A_{i}}{1 + A_{i}\rho(\tau_{k}, \alpha)}$ can easily be shown to be concave with respect to $A_{i}$. Due to the complexity of $f_{i,2}(A_{i})$, it is not easy to show this to be concave. However, as shown in Fig.~\ref{concavity}, numerically $\frac{\partial^{2}f_{i,2}(A_{i})}{\partial A_{i}^{2}} < 0$ over a wide range of system parameters; this suggests that the function is concave for the system parameters considered in this paper. Since linear combinations of concave functions (with positive coefficients) is concave, we present the following conjecture:
\begin{conj}
The probability of coverage in a multi-tier network considering handoff derived in Theorem \ref{propos:CovHandoffMultLay} is concave with respect to $\{A_{i}\}_{i=1}^{K}$.
\end{conj}
\begin{figure}[h!]
\center
\includegraphics [scale = 0.4]{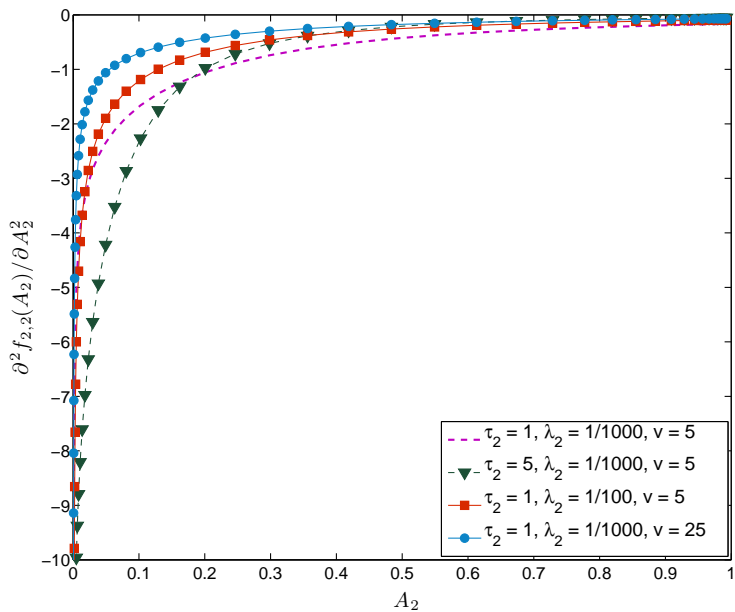}
\caption{The concavity of the term $f_{i,2}(A_{i})$ with respect to $A_{i}$ for the lower tier, i.e., $i = 2$ in a two-tier network .}
\label{concavity}
\end{figure}
The concavity of the objective function, although not leading to a closed-form solution, helps us find the optimum tier association probabilities using standard optimization solvers.

We now show the effect of mobility and more importantly the bias factor on the probability of coverage in a two-tier network through numerical simulations. First, for the purpose of comparison, the probability of coverage for a stationary user is shown in Fig.~\ref{fig:MaxCov}. $\{\lambda_{1}, \lambda_{2}\} = \{0.1, 1\}/(1000m^{2})$ and $\{P_{1}, P_{2}\} = \{46, 20\}$dBm denote the tiers' AP density and transmit power respectively. Tier 1 acts as the reference with bias factor $B_{1} = 1$ and its association probability is given by: $A_{1} = 1 - A_{2}$. As expected, in a two-tier network with equal SIR thresholds, the overall probability of coverage is maximized when the user connects to each tier with equal probability; further, the numerical value of the maximum coverage is independent of the tier AP density or transmit power. This, however, is not the case when mobility and handoff cost is taken into account. In Fig.~\ref{fig:bias}, we consider a two-tier network specified by $\{\lambda_{1}, \lambda_{2}\} = \{0.1, 10\}/(1000 m^{2})$, $\{P_{1}, P_{2}\} = \{46, 20\}$dBm, $\beta = 0.9$ and $\tau_{1} = \tau_{2} = 0$dB, and obtain the optimum tier association, bias factor and the maximum coverage for three different scenarios: 1) ``Optimum Bias" is the solution to~\eqref{OptProb} assuming a concave objective function; 2) ``Optimum Bias at $v = 0$" leads to the optimum tier association for a stationary user regardless of its mobility and handoffs derived in Proposition~\ref{Proposition1}; 3) ``Max-SIR" depicts the scenario where all tiers have the same bias factor, $B_{j} = 1 \; \forall j$, and the user connects to the tier with the maximum average received power. We also compare the obtained results with the optimum solution through a brute force search. As is clear, these results suggest that the conjecture stated above is true for the range of network parameters considered in this paper. The figure illustrates the importance of accounting for handoffs in a multi-tier network. Including the effect of mobility leads to improved coverage by pushing fast-moving users to preferentially pick the higher tiers (tiers with lower densities).

\begin{figure}[h!]
\begin{center}
\includegraphics[width=0.65\textwidth]{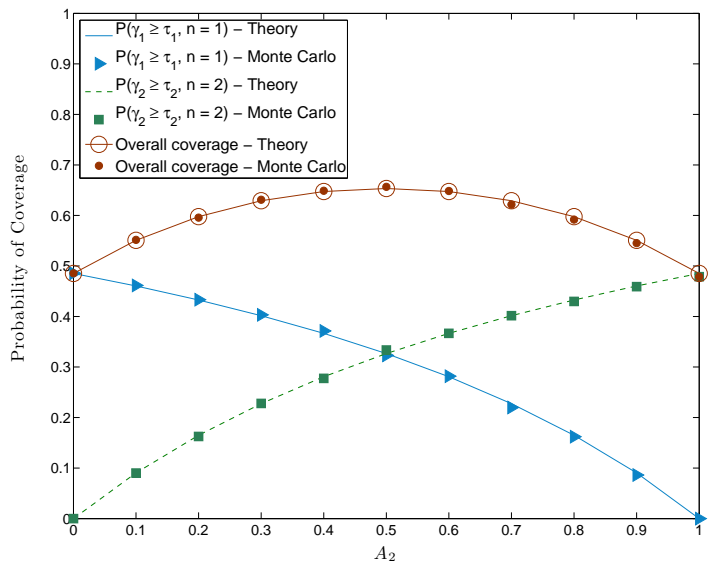}
\caption{Probability of coverage in a two-tier network versus $A_{2}$. $A_{1} = 1 - A_{2},$ $\{\lambda_{1}, \lambda_{2}\} = \{0.1, 1\}/1000$, $\{P_{1}, P_{2}\} = \{46, 20\}$dBm and $\tau_{1} = \tau_{2} = 0$dB. The overall probability of coverage is maximized when $A_{1} = A_{2} = 0.5$.\label{fig:MaxCov}}
\end{center}
\end{figure}

\begin{figure}[tph!]
     \begin{center}
        \subfigure[Overall probability of coverage versus user speed.]{
            \includegraphics[width=0.52\textwidth]{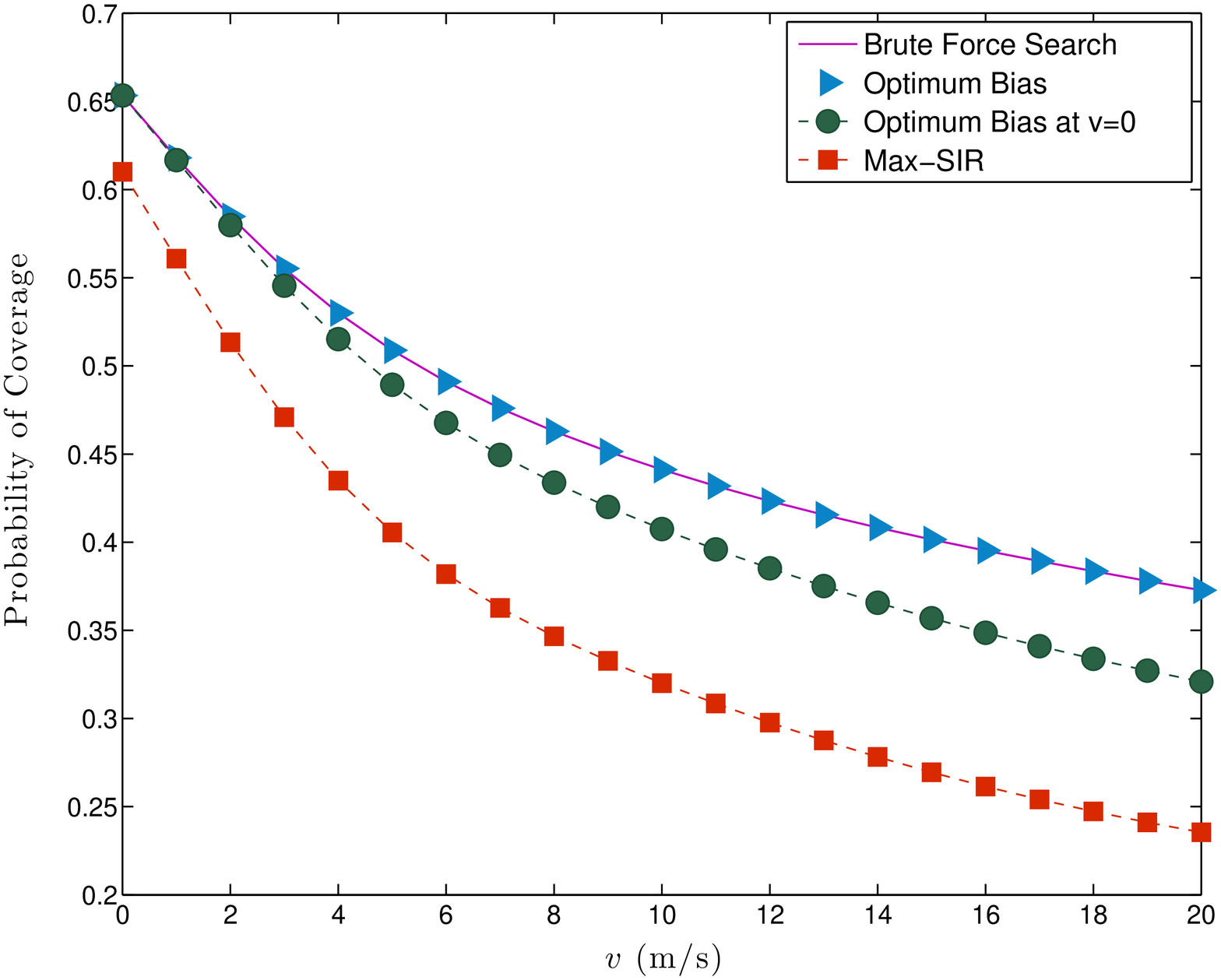}
            \label{Bias1}
        }
        \subfigure[Probability of association to the lower tier versus user speed.]{
           \includegraphics[width=0.52\textwidth]{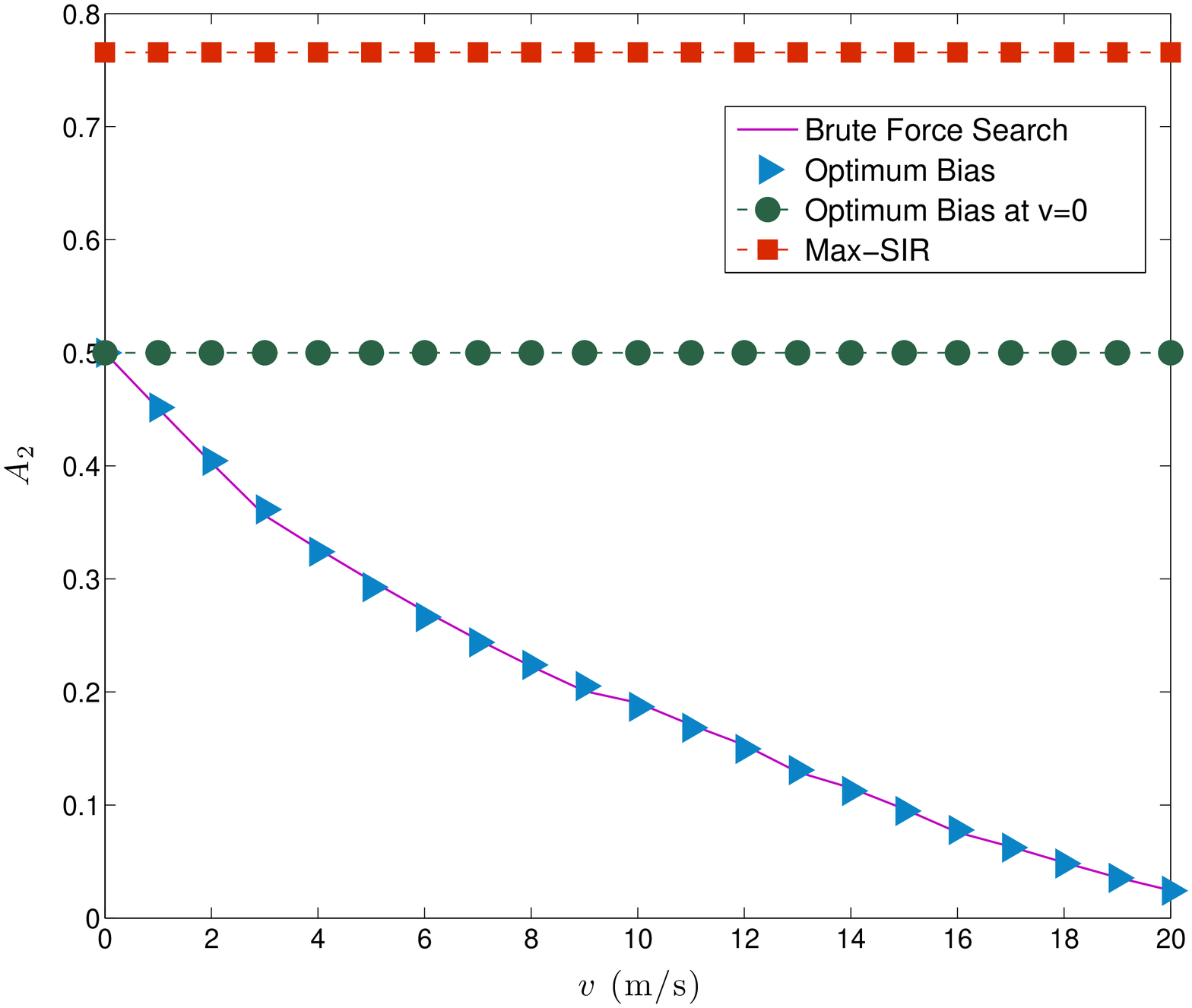}
           \label{Bias2}
         }
          \subfigure[Bias factor for the lower tier. $B_{1} = 1$.]{
           \includegraphics[width=0.52\textwidth]{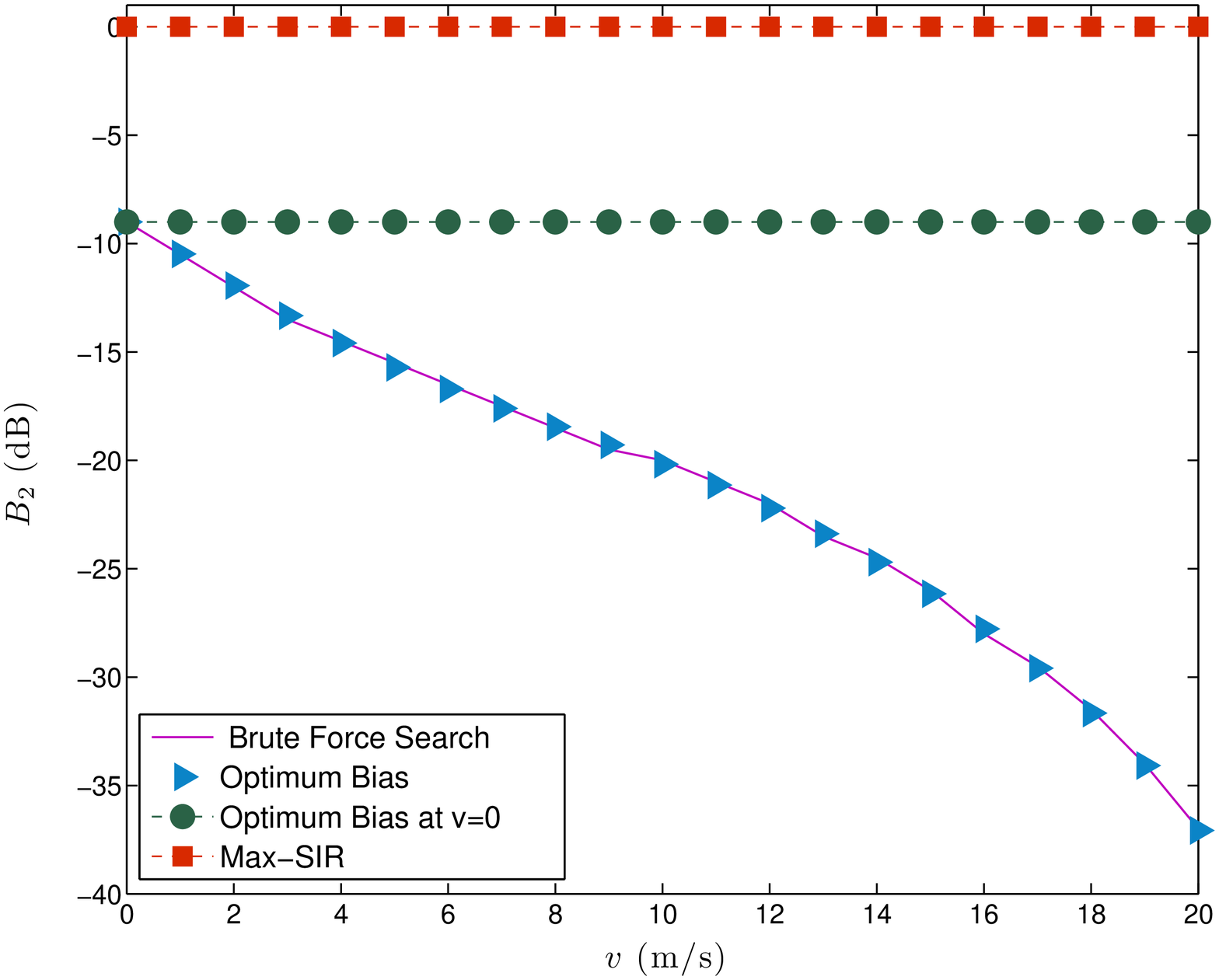}
     	   \label{Bias3}
	 }

        \end{center}
    \caption{
        Coverage in a two-tier network with flexible tier association: (a) overall probability of coverage, (b) probability of association to the lower tier and $\left(c\right)$ the bias factor for the lower tier. $\{\lambda_{1}, \lambda_{2}\} = \{0.1, 10\}/(1000 m^{2})$, $\{P_{1}, P_{2}\} = \{46, 20\}$dBm, $\beta = 0.9$ and $\tau_{1} = \tau_{2} = 0$dB.}
   \label{fig:bias}
\end{figure}


\section{Conclusions}
\label{conclusions}
In this paper, we analyze, and compensate for the impact of user mobility in multi-tier heterogeneous networks. We began by deriving the handoff rate for a typical mobile user in an irregular multi-tier cellular network. The AP locations in each tier are modelled by a homogeneous PPP. We showed that the provided analysis matches the numerical simulations over a broad range of system parameters, i.e., AP density and user speed.

The dependence of the handoff rate on AP density, and the associated cost is the main motivation in assigning users to different tiers of the network based on their velocity. In a multi-tier network where the user connects to an AP of a tier with the maximum biased average received power, the probability of tier association has been shown to be a function of AP transmit power, tier density and bias factor~\cite{shin:12}. While the transmit power and the AP density are mostly determined by the network infrastructure, adjusting the bias factor can change the user association to different tiers in the network. We derived the coverage probability with and without accounting for mobility for the case of orthogonal spectrum allocation among the tiers. We conjecture that optimizing the bias factors is, in fact, a concave problem allowing for efficient solutions. Using the optimal bias factors is shown to improve the coverage probability.

\appendices
\section{Proof of Theorem \ref{thrm:handoffcondition}}
\begin{proof}
\label{AppendixA}
From Fig.~\ref{fig:MobMod1}, for a typical user initially connected to $AP_{s}$ at distance $r$ and moving to the new location $l_{2}$ at distance $R$, a handoff does not occur if there is no other AP closer than $R$ to the user; hence:
\begin{equation} 
\begin{array}{ll}
\label{eq:sets}
1 - \mathds{P}(H_{k} | r, \theta) & = \mathds{P}\left(N\left(\mathcal{|A|}\right) = 1  \; \Big|
 \;  N(|\mathcal{A} \cap\mathcal {C}|) = 1\right)\\
& = \mathds{P}\Big(N(|\mathcal{A} \setminus \mathcal{A} \cap\mathcal {C}|) = 0\Big) \\
& \qeq \exp\Big(-\lambda_{k} (|\mathcal{A} \setminus \mathcal{A} \cap\mathcal {C}|)\Big),
\end{array}
\end{equation}
where $|\cdot|$ denotes the measure of the specified set with $|\varnothing| = 0$, $N(\cdot)$ is the number of APs in the specified area, and $(a)$ results from the null probability of a 2-D Poisson process with density $\lambda_{k}$. The handoff rate depends on the amount of the excess area swiped by the user moving from $l_{1}$ to $l_{2}$ given by:
\begin{equation}
\label{setSubtra}
|\mathcal{A} \setminus \mathcal{A} \cap\mathcal {C}| = |\mathcal{A}| - |\mathcal{A} \cap\mathcal {C}|.
\end{equation}
This measure is the same for the user moving at angle $2\pi - \theta$ with the direction of the connection. Therefore, due to symmetry, we consider $\theta$ being uniformly distributed only in the range of $[0,\pi)$.

In plane geometry, the common area between two intersecting circles with radii $r$ and $R$, where the distance between the centers is $v$, is given by:
\begin{equation} 
\begin{array}{ll}
\label{eq:common}
|\mathcal{A} \cap\mathcal {C}| & = \; \; r ^{2}\cos^{-1}\left ( \frac{r^{2} + v^{2} - R^{2}}{2vr}\right )  + R^{2} \cos^{-1} \left(\frac{R^{2} + v^{2} - r^{2}}{2vR}\right ) \\
& - \; \; \frac{1}{2} \sqrt{(r + R - v)(r + R + v)(v + r - R)(v - r + R)}.\\
\end{array}
\end{equation}
From Fig.~\ref{fig:angels},
\begin{equation}
\label{eq:r&R1}
R^{2} = r^{2} + v^{2} + 2rv\cos \theta,
\end{equation}
and
\begin{equation}
\begin{array}{ll}
\label{eq:r&R2}
r^{2} & =  \; \; R^{2} + v^{2} + 2Rv\cos(\pi - \theta + \phi)\\
& = \; \; R^{2} + v^{2} + 2Rv\cos(\pi - \theta + \sin^{-1}(\frac{v\sin \theta}{R})).
\end{array}
\end{equation}
\begin{figure}
\center
\includegraphics [scale = 0.9]{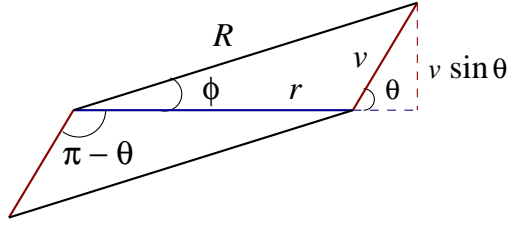}
\caption{Relationship between $r$, $v$ and $R$.}
\label{fig:angels}
\end{figure}
Using~\eqref{eq:r&R1} and~\eqref{eq:r&R2} in~\eqref{eq:common}, the common area is then given by:
\begin{equation}
\begin{array}{ll}
|\mathcal{A} \cap\mathcal {C}| & = \; \; r^{2}\cos^{-1}\left( - \cos \theta\right) + R^{2}\cos^{-1}\left( - \cos(\pi - \theta + \sin^{-1}(\frac{v\sin \theta}{R}) \right)\\
&  - \; \; \frac{1}{2}\sqrt{\left[(r+v) + R\right]\left[(r+v) - R \right]\left[R + (r - v)\right]\left[R - (r - v) \right]},
\end{array}
\end{equation}
where the third term equals:
\begin{equation}
\begin{array}{l}
\displaystyle -\frac{1}{2}\sqrt{\left[(r+v) + R\right]\left[(r+v) - R \right]\left[R + (r - v)\right]\left[R - (r - v) \right]} \\
= \displaystyle -\frac{1}{2} \Big( \left[(r+v)^{2} - R^{2}\right] \left[R^{2} - (r - v)^{2}\right] \Big)^{1/2} \\
= \displaystyle -\frac{1}{2} \Big( 2rv\left[1 - \cos \theta\right] 2rv \left[1 + \cos \theta\right] \Big)^{1/2}  \\
= - rv \sin \theta.
\end{array}
\end{equation}
Using the identity $\cos^{-1}\left(-\cos(\varphi) \right) = \pi - \varphi$, we obtain:
\begin{equation}
\label{finalEq}
|\mathcal{A} \cap\mathcal {C}| = r^{2} (\pi - \theta) + R^{2}\left[\theta - \sin^{-1}\left(\frac{v\sin \theta}{R}\right)\right] - rv\sin \theta.
\end{equation}
Hence, from~\eqref{eq:sets} and~\eqref{setSubtra}, the probability of handoff conditioned on $r$ and $\theta$ is given by:
\begin{equation}
\begin{array}{ll}
\mathds{P}(H_{k} | r, \theta) & = 1 - \exp\left(-\lambda_{k} (|\mathcal{A} \setminus \mathcal{A} \cap\mathcal {C}|)\right) \\
& = 1- \exp \Bigg[ -\lambda_{k} \Bigg(\pi R^{2} - \Big[ r^{2} (\pi - \theta) + R^{2} \left(\theta - \sin^{-1}\left(\frac{v\sin \theta}{R}\right) \right) - rv\sin \theta \Big] \Bigg) \Bigg] \\
& = 1- \exp \Bigg[ -\lambda_{k} \Bigg(R^{2} \Big[ \pi - \theta + \sin^{-1} \left(\frac{v\sin \theta}{R}\right) \Big] - r^{2} (\pi - \theta) + rv\sin \theta \Bigg) \Bigg].
\end{array}
\end{equation}
The handoff rate is then written as:
\begin{multline}
\mathcal{H}_{k}(v, \lambda_{k}) = \mathds {P}(H_{k}) =  \mathds{E}_{\Theta} \Big[ \mathds{E}_{R_{k}} \Big[\mathds{P}(H_{k} | r, \theta) \Big] \Big] \\
= 1 - \frac{1}{\pi} \int_{\theta = 0}^{\pi} \int_{r=0}^{\infty} \exp \Bigg[ -\lambda_{k} \Bigg(R^{2} \left[ \pi - \theta + \sin^{-1}\left(\frac{v\sin \theta}{R}\right) \right] - r^{2}(\pi - \theta) + rv\sin \theta \Bigg) \Bigg] \cdot f_{R_{k}}\left(r\right)  \mathrm{d}r \mathrm{d}\theta \\
= 1 - \frac{1}{\pi} \int_{\theta = 0}^{\pi} \int_{r=0}^{\infty} 2 \pi \lambda_{k} r \exp \Bigg[ -\lambda_{k} \Bigg(R^{2} \left[ \pi - \theta + \sin^{-1}\left(\frac{v\sin \theta}{R}\right) \right] + r^{2}\theta + rv\sin \theta \Bigg) \Bigg]  \mathrm{d}r \mathrm{d}\theta,
\end{multline}
where we used the PDF of $R_{k}$ given by $f_{R_{k}}\left(r\right) = 2 \pi \lambda_{k} r e^{-\pi \lambda_{k}r^{2}}$, and the proof is complete.
\end{proof}

\section{Proof of Theorem \ref{FullRank}}
\label{AppendixB}
\begin{proof}
The proof has two parts. We first show that the determinant of the matrix is given by: $\det$ \textbf{A} = $(-1)^{K-1}\frac{1-\sum_{i=2}^{K}A_{i}^{*}}{\prod_{i = 2}^{K}A_{i}^{*}}$.
The proof is by induction. The statement is true when $K = 2$, since the matrix has only one entry, $1 -  {A_{2}^{*}}^{-1} = \frac{A_{2}^{*} - 1}{A_{2}^{*}}$. When $K > 2$, the coefficient matrix for a $K$-tier network can be written in the form of the block matrix as:
\begin{equation}
\bf{A} = \begin{bmatrix}
 \bf{U} & \bf{V} \\
 \bf{W} & z
\end{bmatrix},
\end{equation}
where $\textbf{U}$ is a square matrix of size $K-2$, $\textbf{V} = [a_{K2}, a_{K3}, \dots, a_{K K-1}]^\top$ is a column vector, $\textbf{W} = [a_{2K}, a_{3K}, \dots, a_{K-1 K}]$ is a row vector and $z =1 - {A_{K}^{*}}^{-1}$ is a scalar. Using determinant of block matrices~\cite{Silvester:00}, we have:
\begin{equation}
\label{detBlock}
\det \textbf{A} = (z - 1) \det \textbf{U} + \det (\textbf{U} - \textbf{V}\textbf{W}), \end{equation} where in the first term $\det$ \textbf{U} = $(-1)^{K-2}\frac{1-\sum_{i=2}^{K-1}A_{i}^{*}}{\prod_{i = 2}^{K-1}A_{i}^{*}}$ by induction. To calculate the second term, note that there is a relation between the off-diagonal entries such that $a_{ij} = 1/a_{ji}$ and $a_{ik}a_{kj} = a_{ij}$ $i \neq j$. Therefore:
\begin{equation}
\textbf{V}\textbf{W} = \begin{bmatrix}
	1  	& a_{32}		& \dots 	& a_{K-1 2} 		\\
	a_{23} 	& 1  		& \dots 	& a_{K-1 3} 		\\
	\vdots 		& \vdots 			& \vdots 	& \vdots			 \\
	a_{2 K-1} 	& a_{3 K-1}		& \dots 	& 1 	\\
	\end{bmatrix}.
\end{equation}
Hence, $\textbf{U} - \textbf{V}\textbf{W}$ is the diagonal matrix $\diag(-{A_{2}^{*}}^{-1}, -{A_{3}^{*}}^{-1}, \cdots, {A_{K-1}^{*}}^{-1})$ with $\det (\textbf{U} - \textbf{V}\textbf{W})$ = $(-1)^{K-2}\prod_{i = 2}^{K-1}{A_{i}^{*}}^{-1}$. Using~\eqref{detBlock} and algebraic manipulation gives the desired result. 

With $\det \textbf{A}$ derived above, the numerator $1-\sum_{i=2}^{K}A_{i}^{*} > 0$, since $\sum_{i = 1}^{K}A_{i}^{*} = 1$, and $A_{i}^{*} \in (0,1)$ for $ i = 1, 2, \dots K$. Hence, \textbf{A} has a non-zero determinant. Therefore, it is a full-rank matrix with rank $K-1$, and the proof is complete.
\end{proof}

\normalem
\clearpage
\bibliographystyle{IEEEtran.bst}
\bibliography{RefV2}
\end{document}